\documentclass[a4paper,UKenglish,cleveref, autoref, thm-restate]{lipics-v2021}

\DeclareMathOperator{\diag}{diag}

\usepackage{thmtools} 
\usepackage[colorinlistoftodos,prependcaption,textsize=tiny,obeyDraft]{todonotes}
\usepackage{microtype}

\usepackage{amsmath,amssymb}
\usepackage{mathtools}
\usepackage{xcolor,pgf,tikz}
\usepackage{comment}
\usepackage{array}
\usepackage{makecell}
\usepackage[ruled,vlined]{algorithm2e}
\usetikzlibrary{arrows,automata,shapes,decorations}

\newtheorem*{lemma*}{Lemma}
\newtheorem*{problem*}{Problem}

\newcommand{\N}{\mathbb{N}}

\newcommand*{\hull}[1]{\mathrm{Co}\left(\kern0pt#1 \right)}
\newcommand*{\cone}[2][{}]{\mathrm{Cone}_{#1}\left(\kern0pt#2 \right)}
\newcommand*{\interior}[1]{ {\kern0pt#1}^{\mathrm{o}}  }
\newcommand*{\closure}[1] { \overline { {\kern0pt#1} } }

\newcommand{\ball}{\mathsf{Ball}}

\newcommand{\R}{\mathbb{R}}

\def\Z{\mathbb{Z}}

\def\C{\mathbb{C}}

\renewcommand{\O}{\mathcal{O}}
\newcommand{\Set}[2]{\left\{ #1 \; \mid \; #2 \right\}}
\newcommand{\clos}[1]{\overline{#1}}
\newcommand{\Q}{\mathbb{Q}}

\newcommand{\bSigma}{\texttt{b-}\Sigma}
\newcommand{\bSigmap}{\bSigma^{++}}
\newcommand{\EA}{\exists\forall_{\leq}\R}
\newcommand{\bEA}{\texttt{b-}\exists\forall_{\leq}\R}
\newcommand{\bEAp}{\texttt{b-}\exists\forall^{++}_{\leq}\R}

\renewcommand{\Re}{\operatorname{Re}}
\renewcommand{\Im}{\operatorname{Im}}
\newcommand{\QFF}{\operatorname{QFF}}

\usepackage{tikz-cd}

\title{On the Complexity of the Escape Problem for Linear Dynamical Systems over Compact Semialgebraic Sets} 

\titlerunning{On the Complexity of the Compact Escape Problem} 
\author{Julian D'Costa}{Oxford University, United Kingdom}{julianrdcosta@gmail.com}{https://orcid.org/0000-0003-2610-5241}{}
\author{Engel Lefaucheux}{Max Planck Institute for Software Systems, Saarland Informatics Campus, Germany}{elefauch@mpi-sws.org}{https://orcid.org/0000-0003-0875-300X}{}
\author{Eike Neumann}{Max Planck Institute for Software Systems, Saarland Informatics Campus, Germany}{eike@mpi-sws.org}{}{}
\author{Joël Ouaknine}{Max Planck Institute for Software Systems, Saarland Informatics Campus, Germany}{joel@mpi-sws.org}{}{}
\author{James Worrell}{Oxford University, United Kingdom}{jbw@cs.ox.ac.uk}{}{}

\authorrunning{J. D'Costa, E. Lefaucheux, E. Neumann, J. Ouaknine, and J. Worrell} 

\Copyright{} 



\category{} 

\supplement{}

\nolinenumbers 

\hideLIPIcs  

\usepackage{float}

\begin{document}
	
	\maketitle

	\begin{abstract}
		We study the computational complexity of the Escape Problem for
	discrete-time linear dynamical systems over compact semialgebraic
	sets, or equivalently the Termination Problem for affine loops with
	compact semialgebraic guard sets. Consider the fragment of the theory
	of the reals consisting of negation-free $\exists \forall$-sentences
	without strict inequalities. We derive several equivalent
	characterisations of the associated complexity class which demonstrate
	its robustness and illustrate its expressive power. We show that the
	Compact Escape Problem is complete for this class. 
	\end{abstract}

	\section{Introduction}

In ambient space $\R^{n}$, a \emph{discrete linear
dynamical system} is an orbit $(X_n)_{n\in\N}$ defined by an
initial vector $X_0$ and a matrix $A$ through the recursion
$X_{n+1}=AX_n$. Linear dynamical systems are fundamental models in
many different domains of science and engineering, and the computability and complexity of
decision problems concerning them are of both theoretical and practical importance.

In the study of dynamical systems, particularly from the perspective
of control theory, considerable attention has been given to the analysis
of \emph{invariant sets}, \emph{i.e.}, subsets of $\R^{n}$ from which
no trajectory can escape; see, \emph{e.g.},
\cite{CastelanH92,BlondelT00,BM07,SDI08}. Our focus in the present
paper is on sets with the dual property that \emph{no trajectory
  remains trapped}. Such sets play a key role in analysing
\emph{liveness} properties:
progress is ensured by
guaranteeing that all trajectories (\emph{i.e.}, from any initial starting
point) must eventually reach a point at which they `escape'
(temporarily or permanently) the set in question, thereby forcing a
system transition to take place.

More precisely, given a rational matrix $A$
and a semialgebraic set
$K \subseteq\R^{n}$, one may consider the \emph{Discrete Escape
	Problem (DEP)} which asks, for all starting points
$X_0$ in $K$, whether the corresponding
orbit of the discrete linear dynamical system $(X_n)_{n\in\N}$
eventually escapes $K$. By ``escaping''
	$K$, we simply mean venturing outside of $K$---we
	are unconcerned whether the trajectory might re-enter
	$K$ at a later time.

      The restriction of DEP to the case in which $K$ is a
      \emph{convex polytope}---alternately known as \emph{termination of linear
        programs} over either the reals or the rationals---was
      already studied and shown decidable in the seminal
      papers~\cite{Tiw04, Braverman2006}, albeit with no complexity
      bounds nor upper bounds on the number of iterations required to
      escape.

In this paper we study the \emph{Compact Escape Problem (CEP)}, a
version of DEP where in addition we assume that the semialgebraic set $K$ is compact. In practice, of
course, this is usually not a burdensome restriction; in most
cyber-physical systems applications, for instance, all relevant
sets will be compact (see, \emph{e.g.},~\cite{Alu15}).
 
CEP was recently shown to be decidable for arbitrary compact
semialgebraic sets in~\cite{NOW20}, via \emph{non-constructive}
methods; consequently---as pointed out in that paper---no non-trivial
complexity bounds could be given. The main contribution of the present
work is to precisely pin down the complexity of CEP in terms of the
first-order theory of the reals; more precisely, we identify a natural
fragment for which CEP is complete. 

Recall that the theory of the reals is concerned with the structure
$\R$ over the signature $\langle\Z, +, \times, \leq, <\rangle$. Tarski
famously showed that this theory is decidable and admits quantifier
elimination, with state-of-the-art techniques based on Collins's
\emph{Cylindrical Algebraic Decomposition} \cite{Collins75} that have
complexity doubly exponential in the number of quantifiers.
Asymptotically faster but arguably impractical quantifier elimination algorithms 
due to \cite{Grigoriev88, HeintzRoySolerno90,Renegar1992} have running time doubly
exponential in the number of quantifier alternations, singly
exponential in the dimension, and polynomial in the rest of the data.
The \emph{existential} fragment of the theory of the reals was famously shown to lie between NP
and PSPACE in~\cite{Canny1988}.




In this paper, we consider the class of formulas consisting of
positive Boolean combinations of non-strict polynomial inequalities
prefixed by a single alternation of a block of existential and a block
of universal quantifiers. Let us denote by $\EA$ the complexity class
of all problems reducible in polynomial time to the decision problem
for this fragment. Using sophisticated results from real algebraic
geometry we show that $\EA$ corresponds to the
decision problem for another fragment of $\exists\forall$-sentences
in which the quantifiers are restricted to range over compact sets, a
result of independent interest. Finally, using techniques from Diophantine
approximation and algebraic number theory we show that the Compact
Escape Problem is complete for this class.

\subsection{Overview}

We formally define the Compact Escape Problem (CEP) as the following decision problem:\\

\emph{Given as input}
	\begin{itemize}
	\item A matrix $A \in \Q^{n \times n}$ with rational entries,
	\item A list $\mathcal{P}$ of polynomials in $\Z[x_1,\dots,x_n]$,
	\item A propositional formula 
		$\Phi(x_1,\dots,x_n)$
		which combines atomic predicates of the form 
		$P(x_1,\dots,x_n) \leq 0$
		with $P \in \mathcal{P}$
		by means of the propositional connectives 
		$\lor$ and $\land$,
	\end{itemize}
\emph{subject to the promise that the set
$K = \Set{x \in \R^n}
{\Phi(x)}$
is compact,
decide whether  for all $x \in K$ there exists $k \in \N$ such that $A^k x \notin K$.}

We assume that the polynomials $P_j$ in the list $\mathcal{P}$ are encoded as lists 
	$\langle (\alpha_{j,k}, c_{j,k})\rangle_{k = 1, \dots, s_j}$
	of pairs of multi-indexes 
	$\alpha_{j,k} \in \N^{n}$,
	whose entries are encoded in unary,
	and coefficients
	$c_{j,k} \in \Z$,
	encoded in binary,
	such that 
	\begin{equation}\label{eq: polynomial standard representation}
	P_j(x_1,\dots,x_n) = \sum_{k = 1}^{s_j} c_{j,k} (x_1,\dots,x_n)^{\alpha_{j,k}}.
	\end{equation}

Note that the analogous problem for affine maps $x \mapsto Ax + b$ reduces to CEP, 
as a point $x \in K$ escapes the compact set $K$ under iterations of the affine map $Ax + b$ if and only if the point 
$(x, 1) \in K \times \{1\}$ escapes $K \times \{1\}$ under iterations of the linear map $B(x,z) = Ax + bz$.

We capture the computational complexity of this decision problem by showing that it is equivalent to the decision problem for a fragment of the theory of the reals. 

Let $\EA$ denote the decision problem for sentences of the form
\begin{equation}\label{eq: intro generic sentence}
	\exists X \in \R^n. \forall Y \in \R^m.
		\left(
			\Phi_{0,\leq}(X,Y)
		\right),
\end{equation}
where $\Phi_{0,\leq}$ is a positive Boolean combination of non-strict polynomial inequalities.
Evidently, this class lies between the existential fragment of the theory of the reals (without restriction on the types of inequalities) and the full $\exists\forall$-fragment.

The main result of this paper is the following:
\begin{theorem}\label{Theorem: first main theorem}
	The compact escape problem is complete for the complexity class $\EA$. 
\end{theorem}

The proof consists of three steps:

First, we show that for any sentence of the form 
\begin{equation}\label{eq: intro generic b-sentence}
	\exists X \in [-1,1]^n. \forall Y \in [-1,1]^m.
		\left(
			\Phi_{0,\leq}(X,Y)
		\right),
\end{equation}
where $\Phi_{0,\leq}$ is a positive Boolean combination of non-strict polynomial inequalities,
one can compute a matrix 
$A \in \Q^{(n + 2m) \times (n + 2m)}$ 
and a compact set 
$K \subseteq \R^{n + 2m}$ 
such that $(A,K)$ is a negative instance of the compact escape problem if and only if \eqref{eq: intro generic b-sentence} holds true.

Secondly, given any instance $(A,K)$ with $A \in \Q^{n \times n}$ and $K \subseteq \R^n$ we can compute in polynomial time a sentence of the form 
\begin{equation}\label{eq: intro generic bp-sentence}
	\exists X \in [-1,1]^{m}. \forall Y \in [-1,1]^{\ell}.
		\left(
			\Psi_{0,\leq}(Y)
			\to 
			\Phi_{0,\leq}(X,Y)
		\right),
\end{equation}
where $\Psi_{0,\leq}$ and $\Phi_{0,\leq}$ are a positive Boolean combination of non-strict polynomial inequalities,
such that \eqref{eq: intro generic bp-sentence} holds true if and only if $(A,K)$ is a negative instance of the compact escape problem.

Finally, we prove that the decision problems for sentences of the form \eqref{eq: intro generic sentence}, \eqref{eq: intro generic b-sentence}, 
and \eqref{eq: intro generic bp-sentence} are all equivalent.	
	
\section{Preliminaries}

\subsection{Fragments of the theory of the reals}

The statement and proof of Theorem~\ref{Theorem: first main theorem} require complexity classes 
induced by decision problems for fragments of the the first-order theory of the reals.
The main goal of this subsection is to formally define these complexity classes.

Thus, let $\mathcal{L}$ be the first-order language with signature
$\langle\Z, +, \times, <, \leq \rangle$,
propositional connectives,
$\land$ and $\lor$,
and quantifiers $\exists$ and $\forall$.
For complexity purposes, we assume that integer constants are encoded in binary. 
See, \textit{e.g.}, \cite{Srivastava, VanDalen} for an introduction to first-order logic.
We interpret all formulas in $\mathcal{L}$ in the structure of real numbers.
Thus, we say that two formulas are equivalent if their interpretations in $\R$ are equivalent.
The restriction to the connectives $\lor$ and $\land$ is of course insubstantial,
and we will make free use of the connectives $\lnot$ and $\rightarrow$ throughout this paper, understanding them as syntactic sugar.

Let $\QFF$ denote the set of quantifier-free formulas in $\mathcal{L}$.
Let $\QFF_{\leq}$ (resp. $\QFF_{<}$) denote the subset of $\QFF$ consisting of those formulas that do not contain the relational symbol ``$<$'' (resp. ``$\leq$'').
Note that the negation of a $\QFF_{\leq}$-formula is a $\QFF_<$-formula and vice versa.

We define the sets of formulas $\Sigma_{n, \leq}$ 
and $\Pi_{n, \leq}$ inductively as follows:
\begin{enumerate}
    \item 
    Let $\Sigma_{0, \leq} = \Pi_{0, \leq} = \QFF_{\leq}$.
    \item 
    A formula $\Psi(y_1,\dots,y_s)$ belongs to $\Sigma_{n + 1, \leq}$
    if and only if it is of the form 
    \[
        \Psi(y_1,\dots,y_s) = (\exists x_1). \dots (\exists x_t). \Phi(x_1,\dots,x_t, y_1,\dots,y_s),
    \]
    where $\Phi$ belongs to $\Pi_{n, \leq}$.
    \item 
    Dually, a formula $\Psi(y_1,\dots,y_s)$ belongs to $\Pi_{n + 1, \leq}$
    if and only if it is of the form 
    \[ 
        \Psi(y_1,\dots,y_s) = (\forall x_1). \dots (\forall x_t). \Phi(x_1,\dots,x_t, y_1,\dots,y_s),
    \]
    where $\Phi$ belongs to $\Sigma_{n, \leq}$.
\end{enumerate}
We define $\Sigma_{n, <}$ and $\Pi_{n, <}$ (resp. $\Sigma_n$ and $\Pi_n$) analogously, starting with $\QFF_<$-formulas (resp. $\QFF$-formulas).

By convention we denote vectors of variables $X = (x_1,\dots,x_t)$ by upper case letters and introduce the shorthand notations
$\exists X$ and $\forall X$ for blocks of quantifiers 
$(\exists x_1). \dots (\exists x_t)$ and $(\forall x_1). \dots (\forall x_t)$.
Recall that a first-order formula $\Phi$ is called a sentence if it does not contain any free variables.

The \emph{decision problem} for a class $\mathcal{C}$ of first-order formulas in the language $\mathcal{L}$ is the following:
Given a sentence that belongs to $\mathcal{C}$ decide whether the sentence holds true in the universe of real numbers.

It is natural to ask how the decision problems for the classes we have introduced 
above are related with respect to polynomial-time reductions.
By taking the negation of formulas it is easy to see that the decision problem for 
$\Sigma_n$ is equivalent to that of $\Pi_n$, 
the decision problem for $\Sigma_{n,\leq}$ is equivalent to that of $\Pi_{n, <}$,
and the decision problem for $\Sigma_{n, <}$ is equivalent to that of $\Pi_{n, \leq}$. As
such it suffices to consider the ``$\Sigma$''-classes in the following.

By a standard trick, any $\QFF$-formula $\Phi(X)$ with free variables $X$ can be converted in polynomial time into an equivalent formula 
$\exists Y. f(X,Y) = 0$ where $f$ is a single polynomial.
It follows that if $n$ is odd then the decision problems for the classes $\Sigma_n$ and $\Sigma_{n, \leq}$ are polynomial-time equivalent
and if $n$ is even then the decision problems for the classes $\Sigma_n$ and $\Sigma_{n, <}$ are polynomial-time equivalent.

Of course, for $n = 0$ the decision problem is trivial for all three classes. 
For $n = 1$ we have the following remarkable result:

\begin{theorem}[{\cite{SS17}}]
    The decision problems for $\Sigma_1$ and $\Sigma_{1, <}$ are polynomial-time equivalent.
\end{theorem}

We thus have polynomial-time reductions for decision problems as indicated below:
\[
    \left(\Sigma_0 \equiv \Sigma_{0,\leq} \equiv \Sigma_{0, <}\right) \rightarrow \left(\Sigma_1 \equiv \Sigma_{1,\leq} \equiv \Sigma_{1, <}\right) \rightarrow \Sigma_{2,\leq} \rightarrow (\Sigma_{2} \equiv \Sigma_{2, <}) \rightarrow \Sigma_{3, <} \rightarrow \dots
\]
It is open to the best of our knowledge whether there exists a reduction of the decision problem for 
$\Sigma_2$
to that of  
$\Sigma_{2,\leq}$.
The techniques from~\cite{SS17} do not seem to carry over to higher orders of quantifier alternations.

We study the decision problem for the class $\Sigma_{2,\leq}$ in greater detail.
Let us denote by $\EA$ the complexity class of all problems reducible in polynomial time to this decision problem.
To demonstrate the robustness of this complexity class and gauge its computational power we give a number of equivalent characterisations.
It turns out that, somewhat surprisingly, the decision problem for $\Sigma_{2,\leq}$-sentences is equivalent to the decision problem for exists-forall-sentences whose quantifiers are restricted to range over compact sets.

Let $X = (x_1,\dots,x_n)$ be a vector of variables.
Let $y$ be a variable or a constant.
We write 
$|X| \leq y$
as an abbreviation for the formula
$
\bigwedge_{j = 1}^n \left(-y \leq x_j \leq y\right)
$.
Of course, this syntactic construct will only have the intended semantics
if our context ensures that $y \geq 0$, and we will only use it in such situations.

Write $I = [-1,1]$.
Let $\Phi_0(X,Y,Z)$ be a quantifier-free formula in $\mathcal{L}$.
We introduce the syntactic abbreviation
\[
    \exists X \in I^n. \forall Y \in I^m.
        \left(
            \Phi_0(X,Y,Z)
        \right)
\]
for the formula
\[
    \exists X \in \R^n. \forall Y \in \R^m.
        \left(
            |Y| > 1
            \lor
            \left(
                |X| \leq 1 
                \land
                \Phi_0(X,Y,Z)
            \right)
        \right)
\]
in the language $\mathcal{L}$.

We have the following result, whose proof is the focus of Section~\ref{sec:proofclasses}:

\begin{theorem}\label{theorem: equivalence of complexity classes}
The decision problems for the following three classes of sentences are equivalent with respect to polynomial-time reduction:
\begin{enumerate}
    \item 
        The class $\Sigma_{2, \leq}$, 
        consisting of sentences of the form
        \[ 
            \exists X \in \R^m. \forall Y \in \R^n. \left(\Phi_{0,\leq}(X, Y)\right),
        \]
        where
        $\Phi_{0,\leq}$ 
        is a 
        $\QFF_{\leq}$-formula.
    \item  
        The class $\bSigma_{2, \leq}$,
        consisting of sentences of the form 
        \[ 
            \exists X \in I^m. \forall Y \in I^n. \left(\Phi_{0,\leq}(X, Y)\right),
        \]
        where
        $\Phi_{0,\leq}$ 
        is a 
        $\QFF_{\leq}$-formula.
    \item 
        The class $\bSigmap_{2,\leq}$,
        consisting of sentences of the form
        \[ 
            \exists X \in I^m. \forall Y \in I^n. \left(\Psi_{0,\leq}(Y) \to \Phi_{0,\leq}(X, Y)\right),
        \]
        where
        $\Phi_{0,\leq}$ 
        and 
        $\Psi_{0,\leq}$
        are 
        $\QFF_{\leq}$-formulas.
\end{enumerate}
\end{theorem}

It is obvious that the decision problem for $\bSigma_{2,\leq}$-sentences reduces to that of 
$\bSigmap_{2,\leq}$-sentences.
Note however that it is not clear that a reduction should exist in either direction between 
$\Sigma_{2,\leq}$ and $\bSigma_{2,\leq}$.
On the one hand, the latter class only allows for quantification over bounded sets, which seems to make it more restrictive.
On the other hand, $\bSigma_{2,\leq}$-sentences involve strict inequalities and hence do not belong to the class $\Sigma_{2,\leq}$.
Let us denote by $\bEA$ and by $\bEAp$ the complexity classes induced respectively by the
decision problem for $\bSigma_{2,\leq}$-sentences and by the decision problem 
for $\bSigmap_{2,\leq}$-sentences.

A remark is in order on the robustness of our definition of the class $\EA$ under different encodings of polynomials.
In practice it is common to encode a polynomial $P$ as a list 
$\langle (\alpha_j, c_{j}) \rangle_{j = 1,\dots, m}$
where $\alpha_j \in \N^{n}$ are multi-indexes and $c_j \in \Z$ are integers satisfying \eqref{eq: polynomial standard representation}.
This is the encoding we have chosen in the definition of CEP. 
By contrast, the polynomials that occur in atomic predicates of a formula in the language $\mathcal{L}$ are encoded as terms over the signature $\langle\Z,+,\times\rangle$.
While one can translate the encoding \eqref{eq: polynomial standard representation} to a term over the signature $\langle\Z,+,\times\rangle$ in polynomial time,
a term of size $N$ can encode a polynomial whose number of non-zero coefficients grows exponentially in $N$, so that a polynomial-time translation in the other direction is not possible in general.
One may hence raise the justified objection that the reduction of CEP to the decision problem for $\Sigma_{2,\leq}$ sentences could hide an exponential overhead in the encoding of the polynomials.
Moreover, in order to show $\EA$-hardness of CEP we need to convert a compact set which is encoded as a $\QFF_{\leq}$-formula into 
an equivalent formula whose atoms use the encoding \eqref{eq: polynomial standard representation}.
We show in Theorem \ref{Theorem: atoms w.l.o.g. quartic} that we can 
efficiently convert any 
$\Sigma_{2,\leq}$-sentence into an equivalent one whose atoms have degree at most $4$.
This resolves the issue, for a uniform bound on the degrees allows one to translate back and forth in polynomial time between the two encodings of polynomials.
While an analogous result for $\Sigma_{2}$-sentences (and, \textit{e.g.}, $\QFF_{\leq}$-formulas) is straightforward 
(see \textit{e.g.} \cite[Lemma 3.2]{SS17} or the proof of Theorem \ref{Theorem: atoms w.l.o.g. quartic} below for a proof idea),
the argument becomes much more involved for $\Sigma_{2,\leq}$-sentences. 
It relies on many of the results that are established in the sequel.
Thus, for the majority of this paper we have to insist on our specific choice of encoding.

\subsection{Mathematical tools}

Our characterisation of the complexity class $\EA$ requires two
sophisticated results from effective real algebraic geometry: Singly
exponential quantifier elimination and a doubly exponential bound on a
ball meeting all components of a semialgebraic set.  We use the
following singly exponential quantifier elimination result given in
\cite{BasuPollackRoy}.  For a historical overview on this type of result
see \cite[Chapter 14, Bibliographical Notes]{BasuPollackRoy}.

\begin{theorem}[{\cite[Theorem 14.16]{BasuPollackRoy}}]\label{Theorem: singly exponential quantifier elimination precise}
    Let $\mathcal{P}$ be a set of at most $s$ polynomials with integer coefficients,
    each of degree at most $d$, 
    in $k + n_1 + \dots + n_{\ell}$ variables.
    Let $\tau$ be a bound on the bitsize of the coefficients of all $P \in \mathcal{P}$.
    Let 
    \[
        \Phi_{\ell}(Y)
        =
        (Q_1 X_1).
        \dots 
        (Q_{\ell} X_{\ell}).
            \left(
                \Psi_0(Y, X_1,\dots, X_{\ell})
            \right),
    \]
    where $Q_j \in \{\exists, \forall\}$ are alternating blocks of quantifiers,
    be a formula over the language $\mathcal{L}$,
    all of whose atoms involve polynomials contained in $\mathcal{P}$.
    Assume that the size of the block of variables $Y$ is $k$ and 
    that the size of the block of variables $X_j$ is $n_j$. 

    Then there exists an equivalent quantifier-free formula
    \[ 
        \omega_0(Y)
        =
        \bigvee_{i = 1}^I \bigwedge_{j = 1}^{J_i} \bigvee_{m = 1}^{M_{i,j}}
            P_{i,j,m}(Y) \bowtie_{i,j,m} 0. 
    \]
    over 
    $\mathcal{L}$,
    where:
    \begin{enumerate}
        \item 
            $I \leq s^{(n_1 + 1)\cdot \dots \cdot (n_\ell + 1)(k + 1)} d^{O(n_1 \cdot \dots \cdot n_{\ell} \cdot k)}$.
        \item $J_i \leq s^{(n_1 + 1)\cdot \dots \cdot (n_\ell + 1)} d^{O(n_1 \cdot \dots \cdot n_{\ell})}$.
        \item $M_{i,j} \leq d^{O(n_1\cdot\dots\cdot n_\ell)}$.
        \item The degrees of the polynomials $P_{i,j,m}$ are bounded by $d^{O(n_1\cdot\dots\cdot n_{\ell})}$.
        \item The bitsize of the coefficients of the polynomials $P_{i,j,m}$ is bounded by $\tau d^{O(n_1\cdot\dots\cdot n_{\ell}\cdot k)}$.
    \end{enumerate}
\end{theorem}

Recall that a \emph{sign condition} on a family 
$\mathcal{P}$ 
of polynomials in $n$ variables is a mapping 
$\sigma \colon \mathcal{P} \to \{-1,0,1\}$.
The \emph{realisation} of a sign condition $\sigma$ in $\R^n$ is the set 
\[ 
    \operatorname{Reali}(\sigma) = \Set{X \in \R^n}{\forall P \in \mathcal{P}.\operatorname{sign}(P(X)) = \sigma(P)}.
\]
A sign condition $\sigma$ is called \emph{realisable} if its realisation is non-empty.
Equivalently, a sign condition is a formula over the language 
$\mathcal{L}$
involving only conjunctions.

The next theorem is due to Vorobjov \cite{Vorobjov84}.
See also \cite[Lemma 9]{GriorievVorobjov88} and \cite[Theorem 4]{BasuRoyRadiusBound}.

\begin{theorem}\label{Theorem: Basu-Roy radius bound}
    There exists an integer constant $\beta'$ with the following property:
    Let $\mathcal{P}$ be a set of $s$ polynomials with integer coefficients in $n$ variables of degree at most $d \geq 2$.
    Assume that the bit-size of the coefficients of each polynomial in $\mathcal{P}$ is at most $\tau$.
    Then there exists a ball centred at the origin of radius at most 
    \[
       2^{\tau d^{\beta' (n + 1)}}
    \]
    which intersects every connected component of every realisable sign condition on 
    $\mathcal{P}$
    in $\R^n$.
\end{theorem}

Our proof of $\EA$-completeness of CEP combines spectral methods with two well-known but nontrivial results on algebraic numbers.
We require a version of Kronecker's theorem on simultaneous Diophantine approximation.
See \cite[Corollary 3.1]{10.5555/2634074.2634101} for a proof.

\begin{theorem}
\label{Theorem: Kronecker}
    Let $(\lambda_1, \dots, \lambda_m)$ be complex algebraic numbers of modulus $1$.
    Consider the free Abelian group
    \[ 
        L = \Set{(n_1,\dots,n_m) \in \Z^m}{\lambda_1^{n_1}\cdot\dots\cdot \lambda_m^{n_m} = 1}.
    \]
    Let $(\beta_1,\dots,\beta_s)$ be a basis of $L$.
    Let $\mathbb{T}^m = \Set{(z_1,\dots,z_m) \in \C^m}{|z_j| = 1}$ denote the complex unit $m$-torus.
    Then the closure of the set 
    $\Set{(\lambda_1^k,\dots,\lambda_m^k) \in \mathbb{T}^m}{k \in \N}$
    is the set 
    $S = \Set{(z_1,\dots,z_m) \in \mathbb{T}^m}{\forall j\leq s. (z_1,\dots,z_m)^{\beta_j} = 1}$.

    Moreover, for all $\varepsilon > 0$ and all $(z_1,\dots,z_m) \in S$ 
    there exist infinitely many indexes $k$ such that 
    $|\lambda_j^k - z_j| < \varepsilon$
    for $j = 1,\dots,n$.
\end{theorem}

Moreover, the integer multiplicative relations between given complex algebraic numbers in the unit circle can be elicited in polynomial time.
For a proof see \cite{CaiLiptonZalcstein00, masser88}.
We assume the standard encoding of algebraic numbers, see \cite{Cohen} for details.

\begin{theorem}\label{Theorem: Masser}
    Let $(\lambda_1, \dots, \lambda_m)$ be complex algebraic numbers of modulus $1$.
    Consider the free Abelian group
    \[ 
        L = \Set{(n_1,\dots,n_m) \in \Z^m}{\lambda_1^{n_1}\cdot\dots\cdot \lambda_m^{n_m}}.
    \]
    Then one can compute in polynomial time a basis $(\beta_1,\dots,\beta_s) \in (\Z^m)^s$ for $L$.
    Moreover, the integer entries of the basis elements $\beta_j$ are bounded polynomially in the 
    size of the encodings of $\lambda_1,\dots,\lambda_m$.
\end{theorem}

\section{Proof of Theorem \ref{theorem: equivalence of complexity classes}}
\label{sec:proofclasses}

Our proof of Theorem \ref{theorem: equivalence of complexity classes} will use 
Theorems \ref{Theorem: singly exponential quantifier elimination precise}
and \ref{Theorem: Basu-Roy radius bound}.
The latter are formulated in terms of the algebraic complexity of a family of polynomials.
We will reformulate them in terms of the bitsize of a formula in the language $\mathcal{L}$.

The \emph{matrix size} $\mu$ of a first-order formula 
\[    
    \Psi(Y)  =  (Q_1 X_1).\dots (Q_{\ell} X_{\ell}) .  \left(        \Phi_0(Y, X_1,\dots, X_{\ell})    \right),
\]
where $Q_j \in \{\exists, \forall\}$
is the number of bits required to write down the quantifier-free part 
$\Phi_0(Y, X_1, \dots, X_\ell)$.
The \emph{dimensions} of the formula $\Psi(Y)$ are the numbers $m, n_1,\dots,n_\ell$, where $m$ is the dimension of $Y$.
The \emph{size} $\sigma$ of the formula $\Psi(Y)$ is the number of bits required to write down the whole formula.
Note that we have $\sigma = O(m + n_1 + \dots + n_{\ell} + \mu)$.

Observe that if $\Phi(X)$ is a $\QFF$-formula of (matrix) size $\mu$ and $P(X) \bowtie 0$ is an atom of $\Phi$
then $P$ has degree at most $\mu$ and its coefficients are bounded in bitsize by $\mu$.
The following is an immediate corollary to Theorem \ref{Theorem: singly exponential quantifier elimination precise}:

\begin{theorem}\label{Theorem: singly exponential quantifier elimination}
    There exists a constant $\alpha$ with the following property:

    Let 
    \[
        (Q_1 X_1). \dots (Q_\ell X_{\ell}).
            \Phi_0(Y,X_1,\dots,X_{\ell})
    \]
    be a first-order formula in the language $\mathcal{L}$
    of matrix size 
    $\mu$
    and with dimensions 
    $m$, $n_1,\dots, n_{\ell}$.
    Then there exists an equivalent quantifier-free formula
    $
        \Psi_0(Y)
    $
    of size at most
    \[
        \mu^{\alpha^{\ell + 1}\left((m + 1) \cdot (n_1 + 1) \cdot \dots \cdot (n_\ell + 1)\right)}.
    \]
\end{theorem}

Theorem \ref{Theorem: Basu-Roy radius bound} entails the following:

\begin{corollary}\label{Corollary: radius bound}
    There exists a constant $\beta$ with the following property:
    Let 
    $\Phi_0(X)$ 
    be a quantifier-free formula in the language $\mathcal{L}$
    of matrix size $\mu$ and dimension $n \geq 1$.
    Then the sentence 
    $
        \exists X \in \R^n.
        \left(
            \Phi_0(X)
        \right)
    $
    is equivalent to the sentence 
    \[
        \exists X.
        \left(
            |X| \leq 2^{\mu^{\beta (n + 1)}}
            \land 
            \Phi_0(X)
        \right)
    \]
\end{corollary}
\begin{proof}
We can write $\Phi_0$ in disjunctive normal form to obtain an equivalent formula 
\[
    \bigvee_{i = 1}^N \left(\bigwedge_{j = 1}^{s_i} P_{i,j}(X) \bowtie_{i,j} 0\right),
\]
with $\bowtie_{i,j} \in \{\leq,<,=\}$.
The atoms $P_{i,j} \bowtie_{i,j} 0$ correspond to atoms of $\Phi_0$.
In particular, each polynomial $P_{i,j}$ has degree at most $\mu$ and coefficients bounded in bitsize by $\mu$.

Now, the sentence
$\exists X.\left(\Phi_0(X)\right)$
is equivalent to the sentence
\[ 
    \bigvee_{i = 1}^N \exists X. 
    \left(
        \bigwedge_{j = 1}^{s_i} P_{i,j}(X) \bowtie_{i,j} 0
    \right).
\]
The latter sentence is, by Theorem \ref{Theorem: Basu-Roy radius bound} equivalent to 
\[ 
    \bigvee_{i = 1}^N \exists X. 
    \left(
        |X| \leq 2^{\mu^{\beta' (n + 1) + 1}}
        \land 
        \bigwedge_{j = 1}^{s_i} P_{i,j}(X) \bowtie_{i,j} 0
    \right).
\]
This is then, by distributivity, equivalent to 
\[ 
\exists X. 
\left(
    |X| \leq 2^{\mu^{\beta' (n + 1) + 1}}
    \land 
    \left(
        \bigvee_{i = 1}^N 
        \bigwedge_{j = 1}^{s_i} P_{i,j}(X) \bowtie_{i,j} 0
    \right)
\right).
\]
which by construction of the disjunctive normal form is equivalent to 
\[ 
    \exists X. 
    \left(
        |X| \leq 2^{\mu^{\beta' (n + 1) + 1}}
        \land 
        \Phi_0(X)
    \right).
\]
The result follows if we let $\beta = \beta' + 1$.
\end{proof}

Theorem \ref{Theorem: singly exponential quantifier elimination} and Corollary \ref{Corollary: radius bound} will allow us to efficiently convert certain formulas into equivalent ones whose quantifiers range over bounded intervals of doubly exponential size in the input data.
By the standard repeated squaring trick such formulas can further be efficiently converted into equivalent ones whose quantifiers range over the interval $I = [-1,1]$:

\begin{lemma}\label{Lemma: removing doubly exponential bounds}
    Given an integer $N$ in unary and a sentence 
    \[
        (Q_1 X_1).
        (Q_2 X_2).
        \dots 
        (Q_s X_s).
            \Phi_{0}(X_1,\dots, X_s),
    \]
    we can in polynomial time in the size of the sentence and 
    $N$ 
    compute a sentence 
    \begin{align*}
        \exists B \in [-1,1]^{N + 1}.
        (Q_1 X_1; |X_1| \leq 1).
        (Q_2 X_2; |X_2| \leq 1).
        \dots 
        (Q_s X_s; &|X_s| \leq 1).\\
        &\Psi_{0}(B,X_1,\dots,X_s)
    \end{align*}
    which is equivalent to the sentence
    \[
        (Q_1 X_1; |X_1| \leq 2^{2^N}).
        (Q_2 X_2; |X_2| \leq 2^{2^N}).
        \dots 
        (Q_s X_s; |X_s| \leq 2^{2^N}).
            \Phi_{0}(X_1,\dots, X_s).
    \]
    Here, the notation 
    $(Q_j; |X_j| \leq c)$ 
    indicates that the quantifier is restricted to 
    the set 
    \[
        \Set{X_j \in \R^{n_j}}{|X_{j,1}| \leq c, \dots, |X_{j,n_j}| \leq c}.
    \]
    Further, if $\Phi_0$ is a $\QFF_{\leq}$-formula then so is $\Psi_0$.
\end{lemma}
\begin{proof}
Introduce fresh variables $b_0,\dots,b_N$.
Let $\Psi_0'$ be the formula that results from $\Phi_0$ by replacing each atom 
\[
    P(X_1,\dots,X_s) \bowtie 0
\]
in $\Phi_0$, where $\bowtie \in \{\leq,<,=\}$,
by the atom
 \[
    b_N^{d_P} \cdot P(X_1 / b_N, \dots, X_s / b_N) \bowtie 0,
\]
where $d_P$ is the total degree of $P$.
Let $\Psi_0$ be the formula 
\[
    \Psi_0'
    \land 
    2b_0 = 1
    \land 
    b_1 = b_0^2
    \dots 
    \land 
    b_{N} = b_{N - 1}^2.
\]
\end{proof}

\subsection{Showing $\EA \subseteq \bEA$}
We now show that the decision problem $\EA$ reduces to $\bEA$ in polynomial time.

We first bound the existential quantifier.
This bound does not yet require the quantifier-free part of the sentence 
to involve only non-strict inequalities.

\begin{lemma}\label{Lemma: bounding existential quantifier}
    Let
    $
    \exists X \in \R^n.
    \forall Y \in \R^m.
        \left(
            \Phi_{0}(X, Y)
        \right).
    $
    be a sentence over the language $\mathcal{L}$
    of matrix size $\mu$.
    Then, denoting $I=[-1,1]$, we can compute in polynomial time an equivalent sentence of the form 
    \[
        \exists X \in I^{n + N}.
        \forall Y \in \R^m.
        \left(
            \Psi_{0}(X,Y)
        \right).
    \]
\end{lemma}
\begin{proof}
    Consider the formula
    $
        \chi_{1}(X) 
        = 
        \forall Y \in \R^m.
            \left(
                \Phi_0(X,Y)
            \right).
    $
    By Theorem \ref{Theorem: singly exponential quantifier elimination}
    this formula is equivalent to a quantifier-free formula 
    $\chi_0(X)$
    of size at most $\mu^{\alpha^2 (n + 1) (m + 1)}$.
    By Corollary \ref{Corollary: radius bound} the sentence
    $
        \exists X \in \R^n.
        \left(
            \chi_0(X)
        \right)
    $
    is equivalent to the sentence
    \[
        \exists X \in \R^n.
        \left(
            |X| \leq 2^{\mu^{\alpha^2 \beta (n + 1)^2 (m + 1)}}
            \land 
            \chi_0(X)
        \right).
    \]
    Hence, our original sentence is equivalent to the sentence 
    \[
        \exists X \in \R^n.
        \forall Y \in \R^m.
        \left(
            |X| \leq 2^{\mu^{\alpha^2 \beta (n + 1)^2 (m + 1)}}
            \land
            \Phi_0(X, Y)
        \right).
    \]
    Now, we can compute in polynomial time a positive integer $N$ in unary such that we have
    $\mu^{\alpha^2 \beta (n + 1)^2 (m + 1)} \leq 2^N$.
    By (the proof of) Lemma \ref{Lemma: removing doubly exponential bounds} we obtain an equivalent sentence as claimed.
\end{proof}

Next we derive a similar bound for the universal quantifier in terms of the bound for the existential one.
This will require the assumption that all inequalities are non-strict.
The reason for this is the following simple continuity property of $\QFF_{<}$-formulas, 
which can fail for general formulas in the language $\mathcal{L}$:

\begin{proposition}\label{Proposition: continuity of QFF<}
    Let $\Phi_0(X)$ be a $\QFF_{<}$-formula with a vector of $n$ free variables $X$.
    Assume that $x \in \R^n$ is such that $\Phi_0(x)$ holds true.
    Then there exists $\varepsilon > 0$ such that $\Phi_0(\tilde{x})$ holds true for all 
    $\tilde{x} \in \R^n$
    with 
    $|x - \tilde{x}| < \varepsilon$.
\end{proposition}
\begin{proof}
    By structural induction on the formula $\Phi$.
    The base case follows from the fact that polynomials are continuous functions. 
    The induction steps are easy.
\end{proof}

\begin{lemma}\label{Lemma: bounding existential quantifier in terms of universal}
    Let $B \in \N$ be a positive integer constant.
    Let
    \[
        \Psi =
        \forall X \in \R^n.
        \exists Y \in \R^m.
        \left(
            |X| > B 
            \lor 
            \Phi_0(X,Y)
        \right)
    \]
    be a $\Pi_{2,<}$-sentence.
    Then the sentence $\Psi$ holds true over the reals if and only if the sentence  
    \[
        \Psi' =
        \exists C \in \R.
        \forall X \in \R^n.
        \exists Y \in \R^m.
        \left(
            |X| > B 
            \lor 
            (
            Y \leq C
            \land 
            \Phi(X,Y)
            )
        \right)
    \]
    holds true over the reals.
\end{lemma}
\begin{proof}
    Clearly, $\Psi'$ implies $\Psi$, so that if $\Psi$ is false then $\Psi'$ is false.

    Suppose now that $\Psi$ is true.
    Let $K = \Set{X \in \R^n}{|X| \leq B}$.
    Then, by assumption,
    for all $X \in K$ there exists $Y(X) \in \R^m$ such that 
    $\Phi(X, Y(X))$ holds true.
    It follows from Proposition \ref{Proposition: continuity of QFF<} 
    that there exists $\varepsilon(X) > 0$ such that 
    $\Phi(X',Y(X))$ holds true for all $X'$ with $|X - X'| < \varepsilon(X)$.
    The set $\Set{\ball(X,\varepsilon(X))}{X \in K}$, where $\ball(X,c)$ denotes the ball
    of radius $c$ centered at $X$,
    is an open cover of $K$.
    The set $K$ is compact, so that this cover has a finite subcover
    $\ball(X_1,\varepsilon(X_1)), \dots, \ball(X_s, \varepsilon(X_s))$.
    It follows that for all $X \in K$ there exists $j \in \{1,\dots,s\}$ such that 
    $\Phi(X, Y(X_j))$ holds true.
    Thus, the formula $\Psi'$ holds true with $C = \max\{|Y(X_1)|,\dots,|Y(X_s)|\}$.
\end{proof}

Note that the conclusion of Lemma \ref{Lemma: bounding existential quantifier in terms of universal} does not hold true in general for $\Pi_{2,\leq}$-formulas.
For instance, the formula 
\[
    \forall x \in [-1,1]. \exists y \in \R. 
    \left( 
        x^2\left(1 - xy\right) \leq 0
    \right)
\]
is clearly true, but the formula
\[
    \exists C \in \R. \forall x \in [-1,1]. \exists y \in [-C,C].
    \left(
        x^2\left(1 - xy\right) \leq 0
    \right)
\]
is clearly false.

\begin{lemma}\label{Lemma: bounding universal quantifier when existential quantifier is bounded}
    Given a sentence of the form 
    \[
        \exists X \in I^n.
        \forall Y \in \R^m.
        \left(
            \Phi_{0,\leq}(X,Y)
        \right),
    \]
    where $\Phi_{0,\leq}$ is a $\QFF_{\leq}$-formula,
    we can compute in polynomial time an equivalent $\bSigma_{\leq}$-sentence
    \[
        \exists X \in I^n.
        \forall Y \in I^{n + M}.
        \left(
            \Psi_{0,\leq}(X,Y)
        \right).
    \]
\end{lemma}
\begin{proof}
    The proof combines Lemma \ref{Lemma: bounding existential quantifier in terms of universal}
    with proof ideas similar to those used in the proof of 
    Lemma \ref{Lemma: bounding existential quantifier}.
    We can compute in polynomial time a sentence 
    \[
        \forall X \in I^n.
        \exists Y \in \R^m.
        \left(
            \chi_{0,<}(X,Y)
        \right),
    \]
    where $\chi_{0,<}$ is a $\QFF_{<}$-formula,
    which is equivalent to the negation of our original sentence.
    By Lemma 
    \ref{Lemma: bounding existential quantifier in terms of universal}
    this sentence is equivalent to the sentence
    \[
        \exists C \in \R. 
        \forall X \in I^n.
        \exists Y \in \R^m.
        \left(
            |Y| \leq C
            \land 
            \chi_{0,<}(X,Y)
        \right).
    \]
    Consider the formula
    \[
        \omega_{2}(C)
        =
        \forall X \in I^n.
        \exists Y \in \R^m.
        \left(
            |Y| \leq C
            \land 
            \chi_{0,<}(X,Y)
        \right).
    \]
    Let $\mu$ denote its matrix size.
    The number $\mu$ is clearly computable in polynomial time from our original sentence.
    By 
    Theorem \ref{Theorem: singly exponential quantifier elimination}
    the formula $\omega_2(C)$ is equivalent to a quantifier-free formula
    $\omega_0(C)$
    of size at most 
    $\mu^{2\alpha^3(n + 1)(m + 1)}$.
    By Corollary \ref{Corollary: radius bound} the sentence
    \[
        \exists C \in \R.
        \left(
            \omega_0(C)
        \right)
    \]
    is equivalent to the sentence
    \[
        \exists C \in \R. 
        \left(
            |C| \leq 2^{\mu^{4\alpha^3\beta(n + 1)(m + 1)}} \land \omega_0(C)
        \right).
    \]
    It follows that the negation of our original sentence is equivalent to the sentence 
    \[
        \exists C \in \R. 
        \forall X \in I^n.
        \exists Y \in \R^m.
        \left(
            |C| \leq 2^{\mu^{4\alpha^3\beta(n + 1)(m + 1)}} 
            \land |Y| \leq C
            \land \chi_{0,<}(X,Y)
        \right).
    \]
    The latter is further equivalent to the sentence
    \[
        \forall X \in I^n.
        \exists Y \in \R^m.
        \left(
            |Y| \leq 2^{\mu^{4\alpha^3\beta(n + 1)(m + 1)}} 
            \land \chi_{0,<}(X,Y)
        \right).
    \]
    Now, compute a positive integer $N$ in unary such that 
    $\mu^{4\alpha^3\beta(n + 1)(m + 1)} \leq 2^N$,
    and proceed as 
    in the proof of Lemma \ref{Lemma: bounding existential quantifier}
    to obtain in polynomial time an equivalent sentence of the form 
    \[
        \forall X \in I^n.
        \exists Y \in I^{m + M}.
        \left(
            \chi_{0,<}(X,Y)
        \right).
    \]
    The result follows by negating this sentence again.
\end{proof}

Lemmas \ref{Lemma: bounding existential quantifier} and 
\ref{Lemma: bounding universal quantifier when existential quantifier is bounded}
together yield the inclusion $\EA \subseteq \bEA$.

\subsection{Showing $\bEA \subseteq \EA$}

We next establish the inclusion
$\bEA \subseteq \EA$.
The key lemma is the following:

\begin{lemma}\label{Lemma: lower bound on epsilon}
    Let 
    \[
        \exists \varepsilon > 0.
        (Q_1 X \in \R^n).
        (Q_2 Y \in \R^m).
        \left(
            \Phi_0(\varepsilon, X, Y)
        \right)
    \] 
    be a sentence over the language $\mathcal{L}$
    of matrix size $\mu$.
    If this sentence holds true, then there exists 
    $\varepsilon > 2^{-\mu^{4\alpha^3\beta(n + 1)(m + 1)}}$
    witnessing the existential quantifier.
\end{lemma}
\begin{proof}
    Consider the formula 
    \[
        \chi_2(\varepsilon) 
        =
        (Q_1 X \in \R^n).
        (Q_2 Y \in \R^m).
        \left(
            \Phi_0(\varepsilon, X, Y)
        \right).
    \] 
    By 
    Theorem \ref{Theorem: singly exponential quantifier elimination}
    this formula is equivalent to a quantifier-free formula
    $\chi_0(\varepsilon)$
    of size at most 
    $\mu^{2\alpha^3(n + 1)(m + 1)}$.
    Let 
    $\chi_0'(\varepsilon)$ 
    be the sentence that results from $\chi_0$ by replacing each atom in
    $P(\varepsilon) \bowtie 0$ in $\chi_0$,
    where $P$ has degree $d$,
    with the atom 
    $\varepsilon^d P(1/\varepsilon) \bowtie 0$.
    Then, evidently, 
    a number
    $\varepsilon > 0$ 
    satisfies 
    $\chi_0(\varepsilon)$ 
    if and only if 
    $1/\varepsilon$
    satisfies 
    $\chi_0'(\varepsilon)$
    and vice versa.

    By Corollary \ref{Corollary: radius bound} 
    the sentence
    $
        \exists x \in \R.\left(x > 0 \land \chi_0'(x)\right)
    $
    is equivalent to the sentence
    \[ 
        \exists x \in \R.\left(x > 0 \land |x| \leq 2^{\mu^{4\alpha^3\beta(n + 1)(m + 1)}} \land \chi_0'(x)\right).
    \]
    The result follows.
\end{proof}

\begin{theorem}\label{Theorem: reduction bounded -> unbounded}
    Given a $\bSigma_{2, \leq}$-sentence 
    \[
        \exists X \in I^n.
        \forall Y \in I^m.
        \left(
            \Phi_{0,\leq}(X,Y)
        \right)
    \]
    we can compute in polynomial time an equivalent
    $\Sigma_{2,\leq}$-sentence.
\end{theorem}
\begin{proof}
    The proof combines 
    Lemma \ref{Lemma: lower bound on epsilon}
    and 
    Proposition \ref{Proposition: continuity of QFF<}
    with similar ideas as in the proof of 
    Lemma \ref{Lemma: bounding existential quantifier}.
    The negation of the sentence is equivalent to a 
    $\Pi_{2, <}$-sentence 
    \begin{equation}\label{eq: bounded -> unbounded eq1}
        \forall X \in I^n.
        \exists Y \in I^m.
        \left(
            \Psi_{0,<}(X,Y)
        \right).
    \end{equation}
    We claim that this sentence 
    is equivalent to the sentence  
    \[
        \exists \varepsilon > 0.
        \forall X \in I^n.
        \exists Y \in (-1 + \varepsilon, 1 - \varepsilon)^m.
        \left(
            \Psi_{0,<}(X,Y)
        \right).
    \]
    Clearly, the latter sentence implies \eqref{eq: bounded -> unbounded eq1}.
    Conversely, assume that \eqref{eq: bounded -> unbounded eq1} holds true.
    Then for all $X \in I^n$ there exists $Y(X) \in I^m$ such that 
    $\Psi_{0,<}(X,Y(X))$ holds true.
    By Proposition \ref{Proposition: continuity of QFF<} there exists 
    for each $X \in I^n$ a number $\varepsilon(X) > 0$ such that the sentence
    $\Psi_{0,<}(\widetilde{X}, \widetilde{Y})$ holds true for all 
    $\widetilde{X}$ and all $\widetilde{Y}$ satisfying
    $|\widetilde{X} - X| < \varepsilon(X)$ 
    and 
    $|\widetilde{Y} - Y(X)| < \varepsilon(X)$.
    Since $I^n$ is compact, the cover 
    $\Set{\ball(X,\varepsilon(X))}{X \in I^n}$
    admits a finite subcover
    $\ball(X_1,\varepsilon_1),\dots,\ball(X_s,\varepsilon_s)$.
    Let $X \in I^n$. 
    Then $X \in \ball(X_j, \varepsilon_j)$ for some $j \in \{1,\dots,s\}$.
    It follows that 
    $\Psi_{0,<}(X, Y)$
    holds true for a
    $Y \in (-1 + \varepsilon_j/2, 1 - \varepsilon_j/2)^m$.
    Thus, the number 
    $\min\{\varepsilon_1/2,\dots,\varepsilon_s/2\}$
    witnesses the existential quantifier in the latter sentence.

    By Lemma \ref{Lemma: lower bound on epsilon} we can compute in polynomial time a positive integer $N \in \N$ in unary 
    such that \eqref{eq: bounded -> unbounded eq1} is equivalent to the sentence 
    \[  
        \forall X \in I^n.
        \exists Y \in \R^m.
        \left( 
            |Y| < 1 - 2^{-2^N}
            \land
            \Psi_{0, <}(X,Y)
        \right).
    \]
    This sentence is further equivalent to the sentence 
    \begin{align*}
        &\forall b_0 \in \R. 
        \dots 
        \forall b_N \in \R. 
        \forall X \in \R^n.
        \exists Y \in \R^{m}.\\
        &\big( 
            \left(
            |X| \leq 1
            \land 
            2b_0 - 1 = 0
            \land 
            b_1 - b_0^2 = 0
            \land
            \dots 
            \land 
            b_N - b_{N - 1}^2 = 0
            \right)\\ 
            &\rightarrow 
            \left(
            |Y| < 1 - b_N
            \land
            \Psi_{0, <}(X,Y)
            \right)
        \big).
    \end{align*}
    This last sentence is a $\Pi_{2,<}$-sentence, so that 
    by negating again we obtain a $\Sigma_{2,\leq}$-sentence equivalent to our original one.
\end{proof}

\subsection{Showing $\bEAp \subseteq \bEA$}
Finally we show the inclusion $\bEAp \subseteq \bEA$.

We will in fact show a stronger but more technical result.
Recall that the Hausdorff distance of two non-empty compact subsets $K$ and $L$ of a metric space $X$
is given by 
\[
    d(K,L) = \max\{\sup_{x \in K} d(x,L), \sup_{x \in L} d(x,K)\},
\] 
where, as usual,
$
    d(x, K) = \inf_{y \in K} d(x,y).
$
This distance function makes the non-empty compact subsets of a metric space into a metric space 
$\mathcal{F}(X)$ of its own.

\begin{theorem}\label{Theorem: adding premise}
    Consider a sentence of the form 
    \[
        \exists X \in I^n.
        \forall Y \in I^m.
        \left(
            \Psi_{0, \leq}(X, Y)
            \rightarrow 
            \Phi_{0, \leq}(X, Y)
        \right),
    \]
    where 
    $\Psi_{0, \leq}(X, Y)$ 
    and 
    $\Phi_{0,\leq}(X, Y)$
    are
    $\QFF_{\leq}$-formulas.
    Assume that the set-valued function
    $
        F(X) = \Set{Y \in I^m}{\Psi_{0,\leq}(X,Y)}
    $
    either maps some $X \in I^n$ to the empty set or 
    is continuous as a map of type $I^n \to \mathcal{F}(I^m)$.
    Then we can compute in polynomial time an equivalent $\bSigma_{2,\leq}$-sentence.
\end{theorem}
\begin{proof}
    See Appendix \ref{Appendix: Proof of adding premise}.
\end{proof}

The inclusion $\bEAp \subseteq \bEA$ follows from the special case of Theorem \ref{Theorem: adding premise} where the formula $\Psi_{0,\leq}(Y)$ does not depend on $X$.

Theorem \ref{Theorem: adding premise}, in its general form, finally allows us to prove that the complexity class $\EA$ is robust under different encodings of polynomials.
We require the following proposition, which is easily established using elementary calculus:

\begin{proposition}\label{Proposition: Hausdorff continuity of products}
    Let $X$ and $Y$ be metric spaces.
    \begin{enumerate}
        \item 
            Let 
            $F \colon X \to \mathcal{F}(Y)$ 
            and 
            $G \colon X \to \mathcal{F}(Z)$
            be continuous with respect to the Hausdorff metric.
            Then the map 
            \[ 
                H \colon X \to \mathcal{F}(Y) \times \mathcal{F}(Z),
                \; 
                H(x) = F(x) \times G(x)
            \]
            is continuous with respect to the Hausdorff metric as well.
        \item 
            Let 
            $F \colon X \to \mathcal{F}(Y)$ be continuous with respect to the Hausdorff metric.
            Let
            $f \colon Y \to Z$ be a continuous function.
            Then the function 
            \[
                H \colon X \to \mathcal{F}(Y \times Z),
                \; 
                H(x) = F(x) \times f(F(x))
            \]
            is continuous with respect to the Hausdorff metric.
    \end{enumerate}
\end{proposition}

\begin{theorem}\label{Theorem: atoms w.l.o.g. quartic}
    Given a $\mathcal{C}$-sentence,
    where $\mathcal{C} \in \{\Sigma_{2,\leq}, \bSigma_{2,\leq}, \bSigma_{2,\leq}^p\}$
    we can compute in polynomial time an equivalent $\mathcal{C}$-sentence whose atoms 
    involve polynomials of degree at most four.
    In particular we can compute in polynomial time a sentence whose atoms involve polynomials encoded 
    as in \eqref{eq: polynomial standard representation}.
\end{theorem}
\begin{proof}
    We prove the result for $\bSigma_{2,\leq}$-sentences.
The result for $\Sigma_{2,\leq}$ sentences follows by applying the reductions from Lemmas \ref{Lemma: bounding existential quantifier} and \ref{Lemma: bounding existential quantifier in terms of universal}, bounding the degrees of the atoms of the resulting $\bSigma_{2,\leq}$-sentence, and translating back to a $\Sigma_{2,\leq}$-sentence using Theorem \ref{Theorem: reduction bounded -> unbounded}. 
By inspecting the proof of Theorem \ref{Theorem: reduction bounded -> unbounded} we observe that the degree does not increase by this translation, since we only add new constraints, all of which involve polynomials of degree at most $2$.
The result for $\bSigmap_{2,\leq}$-sentences is implicitly contained in the below proof.

To a term $T$ over the signature $\langle\Z, +, \times\rangle$ 
we assign a variable $z_T$ and a formula $\eta_T$,
where $\eta_T$ is inductively defined as follows:
\begin{enumerate}
    \item 
        If $T$ is a variable $x_j$ then $\eta_T = \langle z_T = x_j \rangle$.
    \item 
        If $T$ is a constant $c$ then $\eta_T = \langle z_T = c\rangle$.
    \item 
        If $T$ is of the form $U \times V$, then 
        $\eta_T = \langle \eta_U \land \eta_V \land z_T = z_U \times z_V \rangle$
    \item 
        If $T$ is of the form $U + V$, then 
        $\eta_T = \langle \eta_U \land \eta_V \land z_T = z_U + z_V \rangle$.
\end{enumerate}
The formula $\eta_T$ is computable in polynomial time from $T$.
Its atoms have degree at most two.

Let 
$P(X, Y) \leq 0$ 
be an atom in 
$\Phi_{\leq}(X,Y)$,
where $P$ is encoded by a term $T$.
Let $\eta_T$ be the formula associated with $T$ as above.
Then the formula
$P(X, Y) \leq 0$
is equivalent to the formula 
$\forall Z. (\eta_T(X, Y, Z) \to z_T \leq 0)$.

More generally, the sentence 
$\exists X \in I^n. \forall Y \in I^m. \Phi_{\leq}(X,Y)$ 
is equivalent to the sentence 
\[  
    \exists X \in I^n. \forall Y \in I^m.
    \forall Z \in \R^M.
    \left(
        \eta_{T_1}(X,Y,Z) \land \dots \land \eta_{T_s}(X,Y,Z)
            \rightarrow
        \widehat{\Phi}_{\leq}(Z)
    \right),
\]
where 
$T_1,\dots,T_s$
are the term representations of the atoms in
$\Phi_{\leq}(X,Y)$
and
$\widehat{\Phi}_{\leq}(Z)$ 
is obtained from 
$\Phi_{\leq}(X,Y)$ 
by substituting each atom 
$P(X,Y) \leq 0$
with term representation $T_j$
by the atom 
$z_{T_j} \leq 0$.

We can further compute in polynomial time an integer $N$ in binary such that the above sentence is equivalent to 
\[  
    \exists X \in I^n. \forall Y \in I^m.
    \forall Z \in [-N, N]^M.
    \left(
        \eta_{T_1}(X,Y,Z) \land \dots \land \eta_{T_s}(X,Y,Z)
            \rightarrow
        \widehat{\Phi}_{\leq}(Z)
    \right),
\]

By the proof of Lemma \ref{Lemma: removing doubly exponential bounds} we can have $Z$ range over $[-1,1]^M$ up to introducing further auxiliary variables and 
adding a conjunction of quadratic polynomial equations to the formula $\widehat{\Phi}$.
For notational convenience, let us simply assume that the sentence is equivalent to 
\[  
    \exists X \in I^n. \forall Y \in I^m.
    \forall Z \in I^{M}.
    \left(
        \eta_{T_1}(X,Y,Z) \land \dots \land \eta_{T_s}(X,Y,Z)
            \rightarrow
        \widehat{\Phi}_{\leq}(Z)
    \right).
\]
This sentence involves polynomials of degree at most $2$.

Let us write 
$\eta(X,Y,Z) = \bigwedge_{j = 1}^s \eta_{T_j}(X,Y,Z)$.
It remains to show that the set 
\[
    \Set{(Y,Z) \in I^m \times I^M}
        {\eta(X,Y,Z)}
\]
depends continuously on $X$ in the Hausdorff metric.
It then follows from Theorem \ref{Theorem: adding premise} that we can compute in polynomial time an equivalent 
$\Sigma_{2,\leq}$-sentence.
By an inspection of the proof of Theorem \ref{Theorem: adding premise}, the degree of the atoms is at most doubled in this new sentence.

Now,
The formula $\eta$ is a conjunction of atoms of the form 
$
    z_j = x_k
$,
$
    z_j = y_k
$,
$
    z_j = c
$,
$
    z_j = z_{k} + z_{\ell}
$,
or
$
    z_j = z_k \times z_{\ell}
$.

We prove the result by structural induction, using Proposition \ref{Proposition: Hausdorff continuity of products}.
For a formula $\eta(X,Y,Z)$ with $n + m + s$ free variables $(X,Y,Z)$
write $F_{\eta} \colon I^n \to \mathcal{F}(I^{m + s})$ for the map that sends 
$X \in I^n$ 
to the set 
$\Set{(Y,Z) \in I^m \times I^s}{\eta(X,Y,Z)}$. 

If 
$\eta(X,Y,z)$ 
is of the form 
$z = x_k$, $z = y_k$, or $z = c$
then the function 
$F_{\eta}$
is easily seen to be continuous.

If 
$
    \eta(X,Y,z_1,\dots,z_s) = \nu(X, Y, z_1,\dots,z_{s - 1}) \land \mu(X,Y, z_s)
$
where $\mu(X,Y,z_s)$ is of the form 
$z_s = x_k$, $z_s = y_k$, or $z_s = c$
then 
\[
    F_{\eta}(X) 
    = 
    F_{\nu} (X) \times \Set{z_s \in \R}{\mu(X, Y, z_s)}.
\]
Continuity of $F_{\eta}$ follows from the first part of Proposition \ref{Proposition: Hausdorff continuity of products}.

If
$
    \eta(X,Y,z_1,\dots,z_s) = \nu(X, Y, z_1,\dots,z_{s - 1}) \land \mu(X, Y, z_j, z_k, z_s)
$
where 
$\mu(X,Y,z_j,z_k,z_s)$ 
is of the form 
$
    z_s = z_{j} \square z_{s}
$
with $\square \in \{+,\times\}$, then 
\[
    F_{\eta}(X)
    = 
    F_\nu(X) \times f(F_\nu(X)),
\]
where 
$f(Y, z_1,\dots,z_{s - 1}) = z_j \square z_k$.
Continuity of $F_{\eta}$ follows from the second part of Proposition \ref{Proposition: Hausdorff continuity of products}.
\end{proof}

\section{The complexity of deciding the Compact Escape Problem}

We show that CEP is complete for the complexity class
$\EA$.
Formally this is achieved by locating CEP between the complexity classes 
$\bEA$ and $\bEAp$
and applying Theorem \ref{theorem: equivalence of complexity classes}.

Let us first show that CEP is $\EA$-hard.
As a preparation we need to construct in polynomial time an arbitrary finite number of irrational rotations with independent angles:

\begin{lemma}\label{Lemma: rational points without multiplicative relations}
    Given $n \in \N$ in unary we can compute in polynomial time 
    a set of points 
    $q_1,\dots,q_n \in \mathbb{T}^1 \subseteq \C$
    with rational real and imaginary part such that 
    the only integer solution $(e_1,\dots,e_n) \in \Z^n$ to the equation 
    $
        q_1^{e_1}\cdot \dots \cdot q_n^{e_n} = 1
    $
    is the zero vector.
\end{lemma}
\begin{proof}
Recall the following facts about the ring $\Z[i]$ of Gaussian integers, see e.g.~\cite[{Kapitel 1, \S 1}]{Neukirch} for details:
\begin{enumerate}
    \item $\Z[i]$ is a unique factorisation domain.
    \item The units of $\Z[i]$ are $1,-1,i,-i$.
    \item Every prime number $p \in \Z$ with $p \equiv 3 \; (\text{mod } 4)$ is a prime number in $\Z[i]$.
    \item Every prime number $p \in \Z$ with $p \equiv 1 \; (\text{mod } 4)$ admits a factorisation 
    $p = (a + ib)(a - ib)$
    into non-associate prime elements $a + ib, a - ib \in \Z[i]$.
\end{enumerate}
Let $p_1,\dots,p_n$ denote the $n$ first prime numbers with $p_j \equiv 1 \; (\text{mod } 4)$.
By the prime number theorem and a quantitative version of Dirichlet's theorem on primes in arithmetic progressions (see \textit{e.g.} \cite[Chapter 5, Section 3]{BorevichShafarevich} or \cite[Kapitel VII, \S 13]{Neukirch})
there are $\sim \tfrac{N}{2\log N}$ numbers of this type below a given $N \in \N$.
It follows that the numbers $p_1,\dots, p_n$ can be computed in polynomial time from $n$.

Further, we can compute in polynomial time representations 
$
    p_j = a^2_j + b_j^2
$
with $a_j > 0$ for $j = 1,\dots,n$.
Let 
$
    q_j = \tfrac{a^2_j - b^2_j}{a^2_j + b^2_j} + i \tfrac{2a_jb_j}{a^2_j + b_j^2}.
$
We have 
$
    q_j = \tfrac{a_j + ib_j}{a_j - ib_j}
$
where $a_j + ib_j$ and $a_j - ib_j$ are prime elements in $\Z[i]$.

We claim that there are no integer multiplicative relations between the $q_j$'s.
Suppose for the sake of contradiction that we have 
\[
    q_1^{e_1}\cdot \dots \cdot q_{n}^{e_n} = 1
\]
with $e_1,\dots,e_n \in \Z$ not all zero.
Then we obtain the equation 
\[
    (a_1 + ib_1)^{e_1}\cdot \dots (a_n + ib_n)^{e_n}
    =
    (a_1 - ib_1)^{e_1}\cdot \dots (a_n - ib_n)^{e_n}.
\]
Assume without loss of generality that $e_1 \neq 0$.
Then 
$(a_1 - ib_1)$ 
needs to divide one of the prime factors 
$(a_j + ib_j)$.
Since 
$(a_j + ib_j)$ 
is itself prime this implies that
$(a_1 - ib_1)$ 
and 
$(a_j + ib_j)$
are associates.
The units of $\Z[i]$ are the numbers $1, -1, i, -i$.
It follows immediately that the numbers 
$(a_1 - ib_1)$ 
and 
$(a_j + ib_j)$
cannot be associates in $\Z[i]$.
We conclude that there cannot exist any integer multiplicative relations between the $q_j$'s.
\end{proof}

\begin{theorem}
    The Compact Escape Problem is $\EA$-hard.
\end{theorem}
\begin{proof}
    By Theorem \ref{theorem: equivalence of complexity classes} the decision problem for 
    $\bSigma$-sentences is $\EA$-complete.
    It hence suffices to reduce this problem to CEP.

    Thus, given a $\bSigma_{2,\leq}$-sentence 
    $
        \Psi_{2, \leq} 
        =
        \exists x \in I^n.
        \forall y \in I^m.
            \left(
                \Phi_{0,\leq}(x,y)
            \right)
    $
    we compute in polynomial time a compact set $K$ and a rational matrix 
    $A \in \Q^{(n + 2m) \times (n + 2m)}$
    such that there exists a point $x \in K$ with $A^k x \in K$ for all $n \in \N$
    if and only if $\Psi_{2,\leq}$ holds true.

    By Theorem \ref{Theorem: atoms w.l.o.g. quartic} we may assume that all polynomials 
    that occur in $\Psi_{2,\leq}$ have degree at most $4$.

    Consider the compact set 
    \[
        K = \Set{(x,u_1,v_1,\dots,u_m,v_m) \in I^{n} \times I^{2m}}
                {u_j^2 + v_j^2 = 1, \Phi_{0,\leq}(x,u_1,\dots,u_m)}.
    \]
    Use Lemma \ref{Lemma: rational points without multiplicative relations}
    to compute rational numbers 
    $a_1,\dots,a_m, b_1,\dots,b_m \in \Q$
    such that the numbers 
    $a_j + ib_j$
    do not admit any non-trivial integer multiplicative relations.
    Denote by $I_{n}$ the $(n\times n)$-identity matrix.
    Let $R \in \Q^{2m \times 2m}$ be the matrix corresponding to the linear 
    transform which sends a vector $(x_1,y_1,\dots,x_m,y_m) \in \Q^{2m}$
    to the vector 
    \[ 
        (a_1 x_1 - b_1 y_1, b_1 x_1 + a_1 y_1,\dots, a_m x_m - b_m y_m, b_m x_m + a_m y_m).
    \]
    Let $A \in \Q^{(n + 2m) \times (n + 2m)}$ be defined as follows:
    \[
        A = 
        \begin{pmatrix}
            I_{n} &  \\
                           & R
        \end{pmatrix}.
    \]
    Then for all $x \in K$ we have by Theorem \ref{Theorem: Kronecker}
    \[
        \clos{\O_A(x)} = \{x\} \times \Set{(u_1,v_1,\dots,u_m,v_m) \in I^{2m}}{u_j^2 + v_j^2 = 1}.
    \]
    It follows that 
    $\clos{\O_A(x)} \subseteq K$ 
    if and only if 
    $\Phi_{0,\leq}(x,u_1,\dots,u_m)$
    holds true for all $u_1,\dots,u_m \in I^m$.

    Thus, the instance $(A,K)$ of CEP is a negative instance if and only if the sentence 
    $\Psi_{2,\leq}$ holds true.
    We can compute $(A,K)$ in polynomial time from $\Psi_{2,\leq}$.
    This is almost immediately obvious, except that the polynomial inequalities that represent $K$ must be encoded as 
    lists of coefficients, while the polynomial inequalities in $\Psi_{2,\leq}$ are given as terms over the signature
    $\langle \Z, +, \times\rangle$.
    But since the polynomials that occur in $\Psi_{2,\leq}$ have degree at most $4$ we can efficiently compute a list of coefficients from the term representations.
\end{proof}
Conversely, we have:

\begin{theorem}\label{Theorem: contained}
    The Compact Escape Problem is contained in $\EA$.
\end{theorem}
\begin{proof}[Proof Sketch]
    The full proof is given in Appendix \ref{Appendix: proof of contained}.
    We will only briefly sketch the proof idea here.

    Suppose we are given a matrix $A \in \Q^{n \times n}$ with rational entries
    and a family of polynomials $\mathcal{P}$ together with a negation-free propositional formula which encodes a compact set $K \subseteq \R^n$.
    We can compute in polynomial time from this data a $\QFF_{\leq}$-formula $\Phi$ which encodes $K$. 
    We will show that the existence of a point in $K$ that is trapped under $A$ is expressible as a $\bSigmap_{2,\leq}$-sentence.
    Together with Theorem \ref{theorem: equivalence of complexity classes} this yields the result.
    Let us assume for the sake of simplicity that $A$ is diagonalisable over the complex numbers.
    The general case employs the Jordan normal form. 
    It is not more difficult but requires more cumbersome notation.

    We compute the complex eigenvalues 
    $\lambda_1,\dots,\lambda_m,\lambda_{m + 1},\dots,\lambda_{m + b}, \lambda_{m + b + 1},\dots,\lambda_{m + b + s}$ 
    of $A$, counted with multiplicity.
    The eigenvalues are labelled such that 
    $\lambda_1,\dots,\lambda_m$ 
    have modulus $1$,
    such that 
    $\lambda_{m + 1},\dots,\lambda_{m + b}$
    have modulus strictly greater than $1$,
    and such that
    $\lambda_{m + b + 1}, \dots, \lambda_{m + b + s}$
    have modulus strictly smaller than $1$.
    Using \cite{Cai94} we can compute in polynomial time base change matrices $Q$ and $Q^{-1}$ such that 
    $D = Q^{-1} A Q$
    is a diagonal matrix.

    Let $x \in K$ be a starting point.
    If the complex vector $Q^{-1} x$ has a non-zero component $(Q^{-1} x)_j$ with $m + 1 \leq j \leq m + b$ then the orbit of $x$ under $A$ is unbounded, and hence forced to leave the bounded set $K$.

    Now assume that $(Q^{-1} x)_j = 0$ for all $m + 1 \leq j \leq m + b$.
    All components $(Q^{-1}x)_j$ with $j \geq m + b + 1$ converge to zero under the iteration of $A$
    in the sense that the sequence $(Q^{-1} (A^k x))_j$ converges to zero as $k \to \infty$.
    It follows that the closure of the orbit of $x$ under $A$ is equal to the range of the semialgebraic function 
    \[
        f(x,z) = Q \diag\left(z_1,\dots, z_m, 0,\dots,0\right) Q^{-1} x,
    \] 
    where $z_1,\dots,z_m$ range over the closure of the sequence $(\lambda_1^k,\dots,\lambda_m^k)_k$ in the torus $\mathbb{T}^m$.
    By Theorem \ref{Theorem: Kronecker} the closure of this sequence is an algebraic subset of $\mathbb{T}^m$, 
    cut out by the integer multiplicative relations between the eigenvalues $\lambda_1,\dots,\lambda_m$.
    By Theorem \ref{Theorem: Masser} a $\QFF_{\leq}$-formula $\Psi(Z)$ encoding this algebraic set, 
    up to identifying $\mathbb{T}^m$ with a subset of the real hypercube $I^{2m} \subseteq \R^{2m}$.

    It follows that we can express the existence of a trapped point by the following ``informal'' sentence: 
    \begin{align*}
        &\exists X \in I^n.
        \forall Z \in I^{2n}.\\
        &\left(
            \Psi(Z)
           \rightarrow
           \left(
           X \in K
           \land
           \left(
            (Q^{-1} X)_{m + 1} = 0 \land \dots \land (Q^{-1} X)_{m + b} = 0
            \right)
            \land 
            f(X, Z) \in K
            \right)
        \right).
    \end{align*}
    Thanks to the polytime computability of $Q$ and $Q^{-1}$ we can compute in polynomial time formulas that express the relations
    $(Q^{-1} X)_j = 0$ for $j = m + 1,\dots, m + b$,
    and  
    $f(X, Z) \in K$.
    This allows us to compute in polynomial time a $\bSigmap_{\leq}$-sentence which is equivalent to the above ``informal'' sentence.
\end{proof}

	
	\bibliography{refs}
	\appendix
\section{Proof of Theorem \ref{Theorem: adding premise}}
\label{Appendix: Proof of adding premise}

We begin with three simple preparatory observations.

\begin{lemma}\label{Lemma: adding premise simple case}
    Given a sentence of the form  
    \[
        \exists X \in I^n.
        \forall Y \in I^m.
        \left(
            H(X,Y) > 0
        \right),
    \]
    where $H$ is a multivariate polynomial with integer coefficients
    we can compute in polynomial time an equivalent 
    $\bSigma_{2,\leq}$-sentence.
\end{lemma}
\begin{proof}
    The sentence is equivalent to the sentence 
    \[
        \exists \varepsilon > 0.
        \exists X \in I^n.
        \forall Y \in I^m.
        \left(
            H(X,Y) \geq \varepsilon
        \right).
    \]
    By Lemma \ref{Lemma: lower bound on epsilon} this sentence is equivalent to the sentence 
    \[
        \exists \varepsilon \in I.
        \exists X \in I^n.
        \forall Y \in I^m.
        \left(
            \varepsilon \geq 2^{\mu^{4\alpha^3\beta(n + 1)(m + 1)}}
            \land
            H(X,Y) \geq \varepsilon
        \right),
    \]
    where $\mu$ is the size of $h$.
    Compute in polynomial time an integer $N$ such that 
    $\mu^{4\alpha^3\beta(n + 1)(m + 1)} \leq 2^N$
    and apply 
    Lemma \ref{Lemma: removing doubly exponential bounds}
    to obtain the result.
\end{proof}

\begin{lemma}\label{Lemma: bounding polynomial over hypercube}
    Let $P \in \Z[X]$ be a polynomial in $n$ variables,
    encoded by a term $T$ over the signature 
    $\langle \Z, +, \times \rangle$.
    Then we can compute in polynomial time an integer 
    $N$ (in binary) such that 
    $|P(I^n)| \leq N$.
\end{lemma}
\begin{proof}
    We can view $T$ as a tree whose nodes are elements of the set $\{+, \times\}$ and whose leaves are either variables or constants.
    Let $c_1,\dots,c_s \in \Z$ denote the integer constants that occur in $T$.
    Let $M = \max\{2,|c_1|,\dots,|c_s|\}$.

    Let $S$ be the tree which is obtained by substituting $M$ for all leaves in $T$.
    Then $S$ encodes a positive integer $B$.
    This integer $B$ is clearly an upper bound for the absolute value of $P$ over $I^n$.
    By an easy induction argument $B$ is bounded by $M^{N_T}$,
    where $N_T$ is the number of nodes of $T$.
    The number $M^{N_T}$ can be computed using at most $N_T$ arithmetic operations.
    Its bitsize is bounded by $N_T \tau$, where $\tau$ is a bound on the bitsizes of the numbers 
    $c_1,\dots,c_s$.
\end{proof}

\begin{proposition}\label{Proposition: turning propositional formula over equalities into single equality}
    Let $\Phi(X)$ be a quantifier-free formula over the language 
    $\mathcal{L}$ 
    whose atoms consist of equalities only.
    Then we can compute in polynomial time a polynomial $Q \in \Z[X]$
    such that $\Phi(X)$ is equivalent to the formula 
    $Q(X) = 0$.
\end{proposition}
\begin{proof}
    Construct a new formula $\Phi'(X)$ that results from $\Phi(X)$
    by replacing each atom $P(X) = 0$ in $\Phi(X)$ by the atom $P(X)^2 = 0$.

    Now construct a polynomial $Q_{\Phi'}$ by structural induction on $\Phi'$ as follows:
    \begin{enumerate}
        \item 
        If $\Phi'(X) \equiv (P(X) = 0)$ then let $Q_{\Phi'} = P$.
        \item 
        If $\Phi'(X) \equiv \Psi(X) \lor \omega(X)$
        then let 
        $Q_{\Phi'} = Q_{\Psi} \cdot Q_{\omega}$.
        \item 
        If $\Phi'(X) \equiv \Psi(X) \land \omega(X)$ then let 
        $Q_{\Phi'} = Q_{\Psi} + Q_{\omega}$.
    \end{enumerate}
    It is easy to see that $Q_{\Phi'}$ can be computed in polynomial time from $\Phi$.
    It has the desired property by construction.
\end{proof}

We are now in a position to prove Theorem \ref{Theorem: adding premise}.

\begin{proof}[Proof of Theorem \ref{Theorem: adding premise}]
The proof is a reduction to Lemma \ref{Lemma: adding premise simple case}.

As a preparation we assign to every $\QFF_{\leq}$-formula $\Phi$ a continuous function 
$f_{\Phi}$ such that $\Phi(X)$ holds true if and only if $f_{\Phi}(X) \leq 0$:
\begin{enumerate}
    \item If $\Phi(X) = (P(X) \leq 0)$ then let $f_{\Phi}(X) = P(X)$.
    \item If $\Phi(X) = \Psi(X) \lor \chi(X)$ then let 
    $f_{\Phi}(X) = \min\{f_{\Psi}(X), f_{\chi}(X)\}$.
    \item If $\Phi(X) = \Psi(X) \land \chi(X)$ then let 
    $f_{\Phi}(X) = \max\{f_{\Psi}(X), f_{\chi}(X)\}$.
\end{enumerate}

Now assume we are given a sentence 
\begin{equation}\label{eq: adding premise}
        \exists X \in I^n.
        \forall Y \in I^m.
        \left(
            \Psi(X, Y)
            \rightarrow
            \Phi(X, Y)
        \right)
\end{equation}
as above.
The negation of this sentence is equivalent to the sentence
\begin{equation}\label{eq: adding premise negated}
    \forall X \in I^n.
    \exists Y \in I^m.
    \left(
        \Psi(X, Y)
        \land
        f_\Phi(X,Y) > 0
    \right).
\end{equation}
Let us for now assume that the set 
$K(X) = \Set{Y \in I^m}{\Psi(X,Y)}$
is non-empty for all $X \in I^n$.
Then by assumption this set depends continuously on $X$ in the Hausdorff metric.
It follows by elementary calculus that the function 
$
    h(X)
    =
    \max_{Y \in K(X)} f_{\Phi}(X,Y)
$
is well-defined and continuous.

We further have, by compactness of $I^n$, 
that the function $h(X)$ attains its minimum in $I^n$.
By definition of $f_{\Phi}$, 
the sentence \eqref{eq: adding premise negated} holds true if and only if 
$
    \min_{x \in I^n} h(x) > 0
$
if and only if there exists $\varepsilon > 0$ such that 
$
    \min_{x \in I^n} h(x) > \varepsilon
$.
Thus, the sentence \eqref{eq: adding premise negated} is equivalent to the sentence 
\[
    \exists \varepsilon > 0.
    \forall X \in I^n.
    \exists Y \in I^m.
    \left(
        \Psi(X, Y)
        \land
        f_\Phi(X,Y) > \varepsilon
    \right).
\]
So far we have proved this equivalence under the assumption that the compact set $K(X) = \Set{Y \in I^m}{\Psi(X, Y)}$ is non-empty for all $X$.
But if the set $K(X)$ is empty for some $X$ then both \eqref{eq: adding premise negated} and the above sentence are false, 
so that the two sentences are certainly equivalent.

Let $\chi(X,Y)$ be the formula that results from $\Phi$ by swapping all occurrences of $\lor$ and $\land$ and 
by replacing all atoms $P(X,Y) \leq 0$ in $\Phi$ by the atom $P(X,Y) > \varepsilon$.
One easily checks that the above sentence is further equivalent to the sentence 
\[
    \exists \varepsilon > 0.
    \forall X \in I^n.
    \exists Y \in I^m.
    \left(
        \Psi(X, Y)
        \land
        \chi(X,Y)
    \right).
\]
It follows from \ref{Lemma: bounding existential quantifier} that there exists a witness $\varepsilon$ for the existential quantifier
with 
$\varepsilon > 2^{-\mu^{4\alpha^3 \beta(n + 1)(m + 1)}}$.
We can compute in polynomial time an integer $N$ such that we have 
$\mu^{4\alpha^3 \beta(n + 1)(m + 1)} \leq 2^N$.
Consider the formula 
$\chi(X,Y)$.
By Lemma \ref{Lemma: bounding polynomial over hypercube}
we can compute in polynomial time an integer $L$ such that 
$|P(X,Y)| \leq L$ 
for all 
$(X,Y) \in I^n \times I^m$.
We can hence replace each atom 
$P(X, Y) > 0$ 
in 
$\chi(X,Y)$ 
with the equivalent formula 
\[ 
    \exists u \in [-L,L]. \exists v \in [-2^{2^N}, 2^{2^N}]. \left(P(X,Y) = u^2 \land uv = 1\right),
\]
where $u$ and $v$ are fresh variables.
By Proposition \ref{Proposition: turning propositional formula over equalities into single equality}
the formula $\chi(X,Y)$ is equivalent to a formula of the form 
\[
    \exists U \in [-L,L]^s. \exists V \in [-2^{2^N}, 2^{2^N}]^s.
        \left(
            Q(X,Y,U,V) = 0
        \right)
\]
where $Q$ is computable in polynomial time from $\chi(X,Y)$ and $s$ is the number of atoms in $\chi(X,Y)$.

Now, consider the formula $\Psi(X, Y)$.
By Lemma \ref{Lemma: bounding polynomial over hypercube}
we can compute in polynomial time an integer $M$ such that 
for all atoms $P(X, Y) \leq 0$ in $\Psi(Y)$
the polynomial $P$ satisfies $|P(X, Y)| \leq M$ for all 
$(X,Y) \in I^n \times I^m$.
The atom is hence equivalent to 
$\exists w \in [-M,M]. P(X, Y) = -w^2$,
where $w$ is a fresh variable.
Again by Proposition \ref{Proposition: turning propositional formula over equalities into single equality}, 
letting $t$ denote the number of atoms in $\Psi(X, Y)$ we can hence compute in polynomial time 
a formula
$\exists W \in [-M,M]^t. R(X, Y, W) = 0$,
which is equivalent to $\Psi(X, Y)$.

In total the sentence \eqref{eq: adding premise negated} is equivalent to the sentence  
\begin{align*}
   \forall X \in I^n.
    \exists Y \in I^m.
    \exists U \in [-L,L]^s.
    \exists V \in [-2^{2^N}, 2^{2^N}]^s.
    &\exists W \in [-M, M]^t.\\
   &\left(
        R(X, Y, W) + Q(X, Y, U, V) = 0
    \right).
\end{align*}
In the above we have used that the functions $R$ and $Q$ admit only non-negative values by construction.
We may assume that 
$2^{2^N} \geq \max\{L,M\}$.
Arguing as in
Lemma \ref{Lemma: removing doubly exponential bounds}
we can introduce auxiliary variables 
$B \in I^{N + 1}$
to obtain an equivalent sentence 
\[
    \forall X \in I^n.
    \exists Y \in I^m.
    \exists U \in I^s.
    \exists V \in I^s.
    \exists W \in I^t.
    \exists B \in I^{N + 1}.
    \left(
        H(X, Y, U, V, W, B) = 0
    \right)
\]
which is computable in polynomial time from our original sentence \eqref{eq: adding premise}.

The sentence \eqref{eq: adding premise} is hence equivalent to the sentence 
\[
    \exists X \in I^n.
    \forall Y \in I^m.
    \forall U \in I^s.
    \forall V \in I^s.
    \forall W \in I^t.
    \forall B \in I^{N + 1}.
    \left(
        H(X, Y, U, V, W, B) > 0
    \right).
\]
Again, we have used that $H$ only admits non-negative values by construction.
The result now follows from Lemma \ref{Lemma: adding premise simple case}.
\end{proof}
\section{Proof of Theorem \ref{Theorem: contained}}
\label{Appendix: proof of contained}

We start with a technical lemma:

\begin{lemma}\label{Lemma: characterisation of termination}
    Let $A \in \R^{n \times n}$ be a real matrix.
    Denote by 
    \[
        \lambda_1,\dots,\lambda_m, 
        \lambda_{m + 1}, \dots, \lambda_{m + b},
        \lambda_{m + b + 1}, \dots, \lambda_{m + b + s}
    \]
    the complex eigenvalues of $A$, 
    counted with geometric multiplicity.
    Let $\lambda_1,\dots,\lambda_m$ have modulus $1$.
    Let $\lambda_{m + 1}, \dots, \lambda_{m + b}$ have modulus strictly greater than $1$.
    Let $\lambda_{m + b + 1}, \dots, \lambda_{m + b + s}$ have modulus strictly smaller than $1$.
    Fix a Jordan basis 
    $v_{j,k}$
    of 
    $\C^n$
    where 
    $v_{j,1}$
    is an eigenvector of $\lambda_j$
    and 
    $\left(A - \lambda_j I\right)  v_{j,k} = v_{j, k - 1}$
    for all $k > 1$.

    Let $B$ denote the span of the vectors 
    $v_{j,k}$
    with $m + 1 \leq j \leq m + b$
    and the vectors 
    $v_{j, k}$
    with $1 \leq j \leq m$ and $k > 1$.

    Let $C$ denote the span of the vectors 
    $v_{j,k}$
    with $m + b + 1 \leq j \leq m + b$.

    Let $Q$ be the matrix that sends the standard basis of $\C^n$ to the basis 
    \begin{align*}
        &v_{1,1}, \dots, v_{m, 1},\\
        &v_{1, 2}, \dots, v_{1, t_1},
        \dots,
        v_{m, 2}, \dots, v_{m, t_m},\\
        &v_{m + 1, 1},\dots,v_{v_{m + 1}, t_{m + 1}},
        \dots, 
        v_{m + b + s, 1},\dots, v_{m + b + s, t_{m + b + s}}.
    \end{align*}

    Let 
    \[
        f \colon \R^n \times \mathbb{T}^{m} \to \C^n,
        \;
        f(x, z)
        =
        Q
        \begin{pmatrix}
            z_1 &       &     & \\
                &\ddots &     & \\ 
                &       & z_m & \\ 
                &       &     & 0 &  &  \\
                &       &     &   & \ddots & \\
                &       &     &   & & 0 \\
        \end{pmatrix}
        Q^{-1}
        \begin{pmatrix}
            x_1 \\
            \vdots \\
            x_n
        \end{pmatrix}
    \]
    Let 
    $
        S \subseteq \mathbb{T}^m
    $
    be the closure of the set 
    $\Set{(\lambda_1^k, \dots, \lambda_m^k)}{k \in \N}$
    in $\mathbb{T}^m$.

    Let $K \subseteq \R^{n}$ be a compact set.
    Let $x \in K$.
    Then for all $k \in \N$ we have $A^k x \in K$
    if and only if both of the following two conditions are satisfied:
    \begin{enumerate}
        \item 
        Let 
        $N = m + (t_1 - 1) + \dots + (t_m - 1) + t_{m + 1} + \dots + t_{m + b}$.
        For all 
        $m < j \leq N$
        we have 
        $(Q^{-1} x)_{j} = 0$.
        \item 
        $f(x, S) \subseteq K$.
    \end{enumerate}
\end{lemma}
\begin{proof}
    Let $x \in K$.

    Assume that $A^k x \in K$ for all $k \in \N$.
    Let $J = Q^{-1} A Q$.
    Let us again write 
    $N = m + (t_1 - 1) + \dots + (t_m - 1) + t_{m + 1} + \dots + t_{m + b}$.
    If there exists
    $m < j \leq N$ 
    such that 
    $(Q^{-1} x)_j \neq 0$
    then 
    $Q^{-1} x$
    has a non-zero component in a generalised eigenspace of 
    $A$ 
    which corresponds to an eigenvalue of modulus strictly greater than $1$
    or it has a non-zero component in a generalised eigenspace of 
    $A$ 
    corresponding to an eigenvalue of modulus $1$ 
    which is not an eigenspace.
    In both cases the absolute value of 
    $A^k x = Q J^k (Q^{-1} x)$
    is unbounded as $k \to \infty$.
    Since $K$ is assumed to be bounded it follows that $A^k x$ leaves $K$ after finitely many steps.
    
    Now, assume that 
    $(Q^{-1}x)_j = 0$ for all $m < j \leq N$.
    We claim that 
    $f(x, S)$
    is the set of accumulation points of the orbit of $x$ under $A$.
    The result then follows immediately.
    
    First, observe that we have by construction
    \[
        A
        =
        Q
        \begin{pmatrix}
            \lambda_1 &        &            &     &        &   &\\   
                        & \ddots &            &     &        &   &\\
                        &        &\lambda_m &     &        &   &\\
                        &        &            & 0   &        &   &\\
                        &        &            &     & \ddots &   &\\
                        &        &            &     &        & 0 &\\ 
                        &        &            &     &        &   &R \\ 
        \end{pmatrix}
        Q^{-1}
    \]
    where $R$ is an $(s \times s)$-matrix with 
    $|R^k| \to 0$ 
    as 
    $k \to 0$.
    
    Now, let $z \in S$.
    We claim that $f(x,z)$ is an accumulation point of the sequence 
    $(A^k x)_{k \in \N}$.
    Let $\varepsilon > 0$.
    By Theorem \ref{Theorem: Kronecker} there exist infinitely many $k \in \N$ such that 
    $|\lambda_j^k - z_j| < \varepsilon/2$.
    For all sufficiently large $n$ we have $|R^k| < \varepsilon/2$.
    It follows that for each such $k$ we have 
    $
        |(A^k x) - f(z,x)| < \varepsilon.
    $
    Thus, $f(z, x)$ is an accumulation point of the sequence $(A^k x)_k$.
    
    Conversely, let $y \in K$ be an accumulation point of the sequence $(A^k x)_k$.
    Let $(n_k)_k$ be a sequence of natural numbers such that the sequence $(A^{k_j} x)_j$ converges to $y$.
    Since the torus $\mathbb{T}^m$ is compact, the sequence 
    $(\lambda_1^{k_j}, \dots, \lambda_m^{k_j})_j$
    has a convergent subsequence.
    Thus, let $(k_{j_\ell})_\ell$ denote a subsequence of $(k_j)_j$ such that 
    the sequence  
    $(\lambda_1^{k_{j_\ell}}, \dots, \lambda_m^{k_{j_\ell}})_{\ell}$
    converges to a limit 
    $z = (z_1,\dots,z_m) \in \mathbb{T}^m$.
    Then the sequence 
    $(A^{k_{j_\ell}} x)_{\ell}$ 
    converges to both
    $f(x,z)$
    and 
    $y$.
    It follows that $y = f(x,z)$.
\end{proof}

Now, let us prove Theorem \ref{Theorem: contained}.

By Theorem \ref{theorem: equivalence of complexity classes} the decision problems for $\bSigmap_{2,\leq}$-sentences
is contained in $\EA$.
We reduce the Compact Escape Problem to this problem.

Suppose we are given a matrix $A \in \Q^{n \times n}$ with rational entries,
a family of polynomials $\mathcal{P}$ in $n$ free variables, represented in the standard encoding,
and a negation-free propositional formula $\Phi(X)$ over atoms of the form $P \leq 0$, where $P \in \mathcal{P}$.
We can convert the standard encodings of the polynomials $P \in \mathcal{P}$ into terms over the signature $\langle \Z, +, \times \rangle$ in polynomial time.
We can hence convert the formula $\Phi(X)$ into a $\QFF_{\leq}$-formula in polynomial time.
By very slight abuse of notation, let us denote this $\QFF_{\leq}$-formula by $\Phi(X)$ as well.
Let $K \subseteq \R^n$ denote the set encoded by $\Phi(X)$.

By \cite{Cai94} we can compute in polynomial time the complex eigenvalues of $A$
\[
    \lambda_1,\dots,\lambda_m, 
    \lambda_{m + 1}, \dots, \lambda_{m + b},
    \lambda_{m + b + 1}, \dots, \lambda_{m + b + s}
\]
and the matrices $Q$ and $Q^{-1}$ as in Lemma \ref{Lemma: characterisation of termination}.
We can further compute the real an imaginary parts of the eigenvalues 
$\lambda_1,\dots,\lambda_{m + b + s}$
in polynomial time.
More precisely, 
letting $\alpha_j = \Re(\lambda_j)$ denote the real part of $\lambda_j$,
and $\beta_j = \Im(\lambda_j)$ the imaginary part, we can compute in polynomial time:
\begin{enumerate}
\item 
    Univariate polynomials with integer coefficients
    $h_1, \dots, h_{m + b + s}$,
    $g_1, \dots, g_{m + b + s}$,
    such that 
    $h_j(\alpha_j) = g_j(\beta_j) = 0$
    for all 
    $j = 1, \dots, m + b + s$.
\item 
    Rational numbers 
    $a_1, b_1, c_1, d_1, \dots, a_{m + b + s}, b_{m + b + s}, c_{m + b + s}, d_{m + b + s}$,
    such that 
    $\alpha_j$ 
    is the unique root of $h_j$ in the real interval $[a_j, b_j]$
    and 
    $\beta_j$
    is the unique root of $g_j$ in the real interval $[c_j, d_j]$.
\item 
    For 
    $j = 1,\dots,n$ 
    and 
    $k = 1,\dots,n$
    bivariate polynomials
    $L_{0,j,k} \in \Q[u,v]$,
    $L_{1,j,k} \in \Q[u,v]$,
    and indexes 
    $\ell_{j,k} \in \{1,\dots,m + b + s\}$ 
    such that the matrix 
    $Q$ 
    at row $j$ and column $k$
    is given by the complex algebraic number
    \[ 
        L_{0,j,k}(\alpha_{\ell_{j,k}}, \beta_{\ell_{j,k}})
        +
        i L_{1,j,k}(\alpha_{\ell_{j,k}}, \beta_{\ell_{j,k}})
    \]
\item 
    For 
    $j = 1,\dots,n$ 
    and 
    $k = 1,\dots,n$
    bivariate polynomials
    $R_{0, j, k} \in \Q[u, v]$,
    $R_{1, j, k} \in \Q[u, v]$,
    and indexes 
    $r_{j,k} \in \{1,\dots,m + b + s\}$ 
    such that the matrix 
    $R^{-1}$ 
    at row $j$ and column $k$
    is given by the complex algebraic number
    \[ 
        R_{0,j,k}(\alpha_{r_{j,k}}, \beta_{r_{j,k}})
        +
        i R_{1,j,k}(\alpha_{r_{j,k}}, \beta_{r_{j,k}}).
    \]
\end{enumerate}

By Theorem \ref{Theorem: Masser} we can compute in polynomial time a finite set 
$\gamma_1,\dots,\gamma_s \in \Z^m$ 
of generators of the free abelian 
group of integer multiplicative relations between the complex eigenvalues $\lambda_1,\dots,\lambda_m$.
The size of the integer entries of $\gamma_1,\dots,\gamma_s$ -- and not just their bitsize -- is bounded polynomially in the size of the input.
It follows that we can compute in polynomial time a $\QFF_{\leq}$-formula $\Psi(C, D)$ with $2m$ free variables that expresses for two given real vectors 
$C \in \R^n$, $D \in \R^n$ that the complex vector $C + iD$ is contained in the set 
\[
    S = \Set{(z_1,\dots,z_m) \in \mathbb{T}^m}
            {(z_1,\dots,z_m)^{\gamma_j} = 1, \; j = 1,\dots,s}.
\] 
By Theorem \ref{Theorem: Kronecker} the set $S$ is equal to the closure of the set 
$\Set{(\lambda_1^k,\dots,\lambda_m^k)}{k \in \N}$.

Let  
$f \colon \R^n \times \mathbb{T}^m \to \C^n$
be defined as in 
Lemma \ref{Lemma: characterisation of termination},
\textit{i.e.},
\[
    f(x, z)
    =
    Q \diag(z_1,\dots,z_m, 0,\dots,0) Q^{-1} x.
\]
Since we can compute the matrices $Q$ and $Q^{-1}$ in polynomial time as above,
we can compute in polynomial time polynomials
$F_{k, j} \in \Q[U, V][C, D]$
for 
$k = 1,\dots,n$, $j = 1,\dots,n$,
where $U$ and $V$ are vectors of $m + b + s$ variables,
such that 
\begin{equation}\label{eq: real part of f}
    \Re f(X, C + i D) = 
        \left(
            \sum_{j = 1}^n
            F_{1,j}
            \left(
                \vec{\alpha},
                \vec{\beta}
            \right)
            \left(
                C,
                D
            \right) \cdot X_j,
            \dots, 
            \sum_{j = 1}^n
            F_{n,j}
            \left(
                \vec{\alpha},
                \vec{\beta}
                \right)
                \left(
                C,
                D
            \right)  \cdot X_j
        \right).
\end{equation}
Note that the result is a polynomial with real algebraic coefficients.
More precisely, the right hand side of the above equation is an element of the ring 
\[ 
    \Q[\alpha_1, \dots, \alpha_{m + b + s}, \beta_1, \dots, \beta_{m + b + s}][X, C, D].
\]

Define 
$N = m + (t_1 - 1) + \dots + (t_m - 1) + t_{m + 1} + \dots + t_{m + b}$
as in Lemma \ref{Lemma: characterisation of termination}.
By Lemma \ref{Lemma: characterisation of termination} the existence of a point in $K$ that is trapped under $A$ is equivalent to the ``informal'' sentence 
\begin{align}\label{eq: inclusion informal sentence}
    &\exists X \in I^n.
    \forall Y \in \mathbb{T}.\\
    &\left(
        Y \in S
        \rightarrow
        \left(
        X \in K
        \land
        \left(
        (Q^{-1} X)_{m + 1} = 0 \land \dots \land (Q^{-1} X)_N = 0
        \right)
        \land 
        f(X, Y) \in K
        \right)
    \right)\nonumber.
\end{align}

We construct in polynomial time from $A$ and $\Phi$ a $\bSigmap_{\leq}$-sentence
\begin{align}\label{eq: inclusion formal sentence}
    &\exists U \in I^{m + b + s}.
    \exists V \in I^{m + b + s}.
    \exists X \in I^{n}.
    \forall C \in I^{m}.
    \forall D \in I^{m}.\\
    &\left(
        \Psi(C, D)
        \rightarrow
        \left(
            \chi(U,V)
            \land 
            \Phi(X)
            \land 
            \omega(U, V, X)
            \land 
            \xi(U, V, X, C, D)
        \right)
    \right)\nonumber.
\end{align}
Recall that the formula $\Psi(C, D)$ expresses that 
the complex number 
$C + iD$ 
is contained in the set $S$.
Intuitively speaking,
the formula 
$\chi(U,V)$ 
will express that the variables 
$U$ and $V$ represent
the real and imaginary parts of the eigenvalues 
$\lambda_1,\dots,\lambda_{m + b + s}$.
The formula 
$\omega(U, V, X)$
will express that 
$(Q^{-1} X)_k = 0$
for  
$k = m + 1,\dots, m + b + s$.
The formula 
$\xi(U, V, X, C, D)$
will express that 
$f(X , C + i D) \in K$.

More formally, let 
\[
    \chi(U,V) 
    =
    \bigwedge_{j = 1}^{m + b + s}
        \left(
            h_j(U) = 0 \land a_j \leq U \leq b_j
            \land 
            g_j(V) = 0 \land c_j \leq V \leq d_j
        \right).
\]
Let 
\[
    \omega(U, V, X)
    =
    \bigwedge_{k = m + 1}^{N}
    \bigwedge_{s = 0}^1
        \left(
            \sum_{j = 1}^n
                R_{s, k, j}(U_{r_{j,k}}, V_{r_{j,k}}) \cdot X_j = 0
        \right),
\]
Let 
$\xi(U, V, X, C, D)$
be the formula which is obtained from 
$\Phi$
by replacing each atom 
$P(X_1,\dots,X_n) \leq 0$
in 
$\Phi$
by the atom 
\[
    P
    \left(
        \sum_{j = 1}^{n} F_{1,j}(U,V)(C, D) \cdot X_j,
        \dots,
        \sum_{j = 1}^n F_{n,j}(U,V)(C, D) \cdot X_j
    \right)
    \leq 
    0,
\]
Note that this substitution can be performed in polynomial time.
The polynomial $P$ is given by a term $t$ over the signature 
$\left\langle \Z, +, \times\right\rangle$.
A term representing the new atom is obtained by substituting in the term $t$
the occurrence of each variable $X_k$ by the polynomial-size term 
$\sum_{j = 1}^n F_{k,j}(U,V)(C, D) \cdot X_j$. 

Now, observing that the formula $\chi(U,V)$ forces $U$ and $V$ 
to be equal respectively to the vector of real and imaginary parts of the eigenvalues 
$\lambda_1,\dots,\lambda_{m + b + s}$
it follows by construction that the 
$\bSigmap_{\leq}$-sentence \eqref{eq: inclusion formal sentence}
is equivalent to the informal sentence \eqref{eq: inclusion informal sentence}
and hence expresses the existence of a trapped point.
There is only one small argument required:
By \eqref{eq: real part of f} the formula 
$\xi(\vec{\alpha},\vec{\beta}, X, C, D)$
expresses that 
$\Re f(X, C + i D) \in K$
rather than 
$f(X, C + i D) \in K$.
But if 
$\Psi(C,D)$
holds true then $C + i D \in S$,
so that $f(X, C + iD)$ is real-valued,
for instance since it is contained in the closure 
of the orbit of $A^k x$ by the proof of 
Lemma \ref{Lemma: characterisation of termination}.

Deciding the truth of the sentence \eqref{eq: inclusion formal sentence} is therefore equivalent to deciding non-termination of the 
Escape Problem instance $(A,K)$.

\end{document}